\newif\ifarxiv

\arxivtrue

\ifarxiv
\documentclass{article}
\usepackage{arxiv}
\usepackage{amsthm}
\else
\documentclass[10pt,twocolumn,twoside]{IEEEtran}

\def\QEDclosed{\mbox{\rule[0pt]{1.3ex}{1.3ex}}} 

\def\QED{\QEDclosed} 

\fi

\usepackage{generic}
\usepackage{cite}
\usepackage{amsmath,amssymb,amsfonts,todonotes}
\usepackage{algorithmic}
\usepackage{graphicx}
\usepackage{textcomp}
\usepackage{url}
\usepackage{wrapfig}

\newtheorem{theorem}{Theorem}
\newtheorem{lemma}{Lemma}

\newtheorem{remark}{Remark}
\newtheorem{definition}{Definition}

\newtheorem{assumption}{Assumption}

\usepackage{bbm}
\usepackage[normalem]{ulem}

\newcommand{\R}{\mathbb{R}}

\newcommand{\N}{\mathcal{N}}
\newcommand{\A}{\mathcal{A}}
\newcommand{\X}{\mathcal{X}}

\newcommand{\E}{\mathbb{E}}
\newcommand{\xbar}{\bar{x}}
\newcommand{\mm}{\mathcal{M}}

\newcommand{\Tb}{\mathbf{T}}

\newcommand{\abold}{\bm{\alpha}}
\newcommand{\Vhat}{\widehat{V}}
\newcommand{\Jhat}{\widehat{J}}

\newcommand{\pibar}{\bar{\pi}}

\newcommand{\Vbar}{\overline{V}}

\DeclareMathOperator*{\argmax}{arg\,max}
\usepackage[ruled,vlined,linesnumbered ]{algorithm2e}
\usepackage{bm}
\newcommand{\abar}{\bar{\abold}}
\newcommand{\rbar}{\bar{r}}
\newcommand{\rmax}{r_{\text{max}}}
\newcommand{\muhat}{\widehat{\mu}}
\newcommand{\xnot}{\{x\}_{-\N_W}}
\newcommand{\Jmaxt}{\tilde{J}_{\max}}
\newcommand{\blue}[1]{\textcolor{black}{#1}}
\newcommand{\bluetoo}[1]{\textcolor{black}{#1}}
\newcommand{\bluethree}[1]{\textcolor{black}{#1}}

\begin{document}
\title{
Model-Free Learning and Optimal Policy Design in Multi-Agent MDPs Under Probabilistic Agent Dropout

}


\author{Carmel Fiscko$^{1}$, Soummya Kar$^{2}$, and Bruno Sinopoli$^{3}$
\thanks{$^{1}$Carmel Fiscko is with the Dept. of Electrical and Computer Engineering at Carnegie Mellon University at 5000 Forbes Ave, Pittsburgh, PA 15213 and the Dept. of Systems Engineering at Cornell University in Ithaca, NY. {\tt\small cfiscko@cornell.edu}. This material is based upon work supported by the National Science Foundation Graduate Research Fellowship Program under Grant No. DGE1745016. Any opinions, findings, and conclusions or recommendations expressed in this material are those of the author(s) and do not necessarily reflect the views of the National Science Foundation. Additional support from the Hsu Chang Memorial Fellowship in ECE.}%
\thanks{$^{2}$Soummya Kar is with the Dept. of Electrical and Computer Engineering at Carnegie Mellon University at 5000 Forbes Ave, Pittsburgh, PA 15213. {\tt\small soummyak@andrew.cmu.edu}}%
\thanks{$^{3}$Bruno Sinopoli is with the Dept. of Electrical and Systems Engineering at Washington University in St. Louis, MO at 1 Brookings Dr, St. Louis, MO 63130. {\tt\small bsinopoli@wustl.edu }}%
}

\maketitle

\begin{abstract}
This work studies a multi-agent Markov decision process (MDP) that can undergo agent dropout and the computation of policies for the post-dropout system based on control and sampling of the pre-dropout system. The \bluethree{central planner}'s objective is to find an optimal policy that maximizes the value of the expected system given a priori knowledge of the agents' dropout probabilities. For MDPs with a certain transition independence and reward separability structure, we assume that removing agents from the system forms a new MDP comprised of the remaining agents with new state and action spaces, transition dynamics that marginalize the removed agents, and rewards that are independent of the removed agents. We first show that under these assumptions, the value of the expected post-dropout system can be represented by a single MDP; this ``robust MDP" eliminates the need to evaluate all $2^N$ realizations of the system, where $N$ denotes the number of agents. More significantly, in a model-free context, it is shown that the robust MDP value can be estimated with samples generated by the pre-dropout system, meaning that robust policies can be found before dropout occurs. This fact is used to propose a policy importance sampling (IS) routine that performs policy evaluation for dropout scenarios while controlling the existing system with good pre-dropout policies. The policy IS routine produces value estimates for both the robust MDP and specific post-dropout system realizations and is justified with exponential confidence bounds. Finally, the utility of this approach is verified in simulation, showing how structural properties of agent dropout can help a controller find good post-dropout policies before dropout occurs.

\end{abstract}

\section{Introduction}
\label{sec:introduction}
\bluethree{Research in the modeling and control of multi-agent systems (MAS) attempts to describe how the decisions of individuals translates into group behavior. For example, consider social media interactions \cite{le2022socialbots}, energy management of microgrids \cite{fang2021multi}, and cooperative control of autonomous vehicles \cite{chen2021graph}. Related problems include consensus formation \cite{zhu2023neural}, adversarial agents \cite{zhang2020robust}, and the related computational challenges as the number of agents grows \cite{carmona2023model}. Regulators who aim to achieve objectives on the system must therefore understand these interactions in designing satisfactory control policies.}

\bluethree{In this work we consider hierarchical control of MAS behavior where the agents' decision processes are affected by a central planner (CP). The agents select actions in pursuit of local objectives, but these decisions are impacted by a global control, such as a bandwidth budget or pricing schemes. The CP selects actions with the purpose of corralling the agents' choices and thus controlling the resulting stochastic process to a desired outcome. If the objective of the CP is aligned with those of the agents, then this formulation reduces to a centralized solution approach to find local policies. To achieve this goal, the CP considers their control capabilities, the agents' actions, and the problem  limitations in terms of visibility, data, and computation ability. }

\bluethree{A Markov decision process (MDP) model may be constructed whereby the state space describes the agents' behavior, the action space is the CP's controls, and the reward function encodes the CP's control objectives. The state-action to state transitions describe the agents' local decision-making, which may be known from a given model or may be learned from prior reinforcement learning (RL) or multi-agent reinforcement learning (MARL) implementations \cite{shapley1953stochastic} \cite{yang2020overview}. Given the MDP model, the CP can solve for a global policy to achieve their control objectives \cite{Fiscko2019ControlOP}. If each agent samples their next state independently given the current state, the MDP model can be expressed as a factored \cite{osband2014near} or transition-independent MDP \cite{becker2003transition}, which reduces the problem scale \cite{fiscko2022cluster}.}

Often, theoretical guarantees on policy performance require the MDP model to be stationary, which is not reasonable in practice.  One solution is to learn system parameters such as the transition matrix or Q values and update them over time; however, this time-delayed model induces lag and lacks strong theoretical guarantees. Recent investigations have assumed the MDP is stationary within discrete time units, identified the system for each block of time, and chosen a policy based on the temporary model \cite{ornik2019learning} \cite{lecarpentier2019non} \cite{cheung2020reinforcement}.

\subsection{Agent Dropout}
\bluethree{In this paper, we focus on a specific type of non-stationarity: removal of agents from the MAS. Agent dropout occurs when, after some time has passed under normal operation, an agent or group of agents leaves the system. This changes the fundamental structure of the MAS, changing all aspects of the MDP model. A control policy based on the full set of participating agents may not yield good value if dropout occurs; in the worst case, the policy may induce unsafe system operation. For example, consider a coverage problem in which a centrally controlled swarm of drones  maintains a formation where each drone is equidistant from one another. If a drone were to lose power or connectivity, the CP would need to react to maintain the same objective for the reduced number of agents. Versions of this problem in terms of the graph edges have been studied in \cite{summers2009addressing} and \cite{gasparri2017bounded}. If dropout occurs and the CP uses Q-learning, it is likely that finding a tolerable post-dropout policy will take too long.}

\bluethree{Agent dropout has previously been studied as link failure within a communication graph. For example, in distributed settings, update protocols have been developed that are robust to stochastic networks \cite{kar2013cal}, \cite{zhang2018fully}. These papers necessitate the assumption is that the graph is fully connected on average, thus ensuring convergence of the agents' value estimates. This work focuses on a related but different formulation: we make no such connectivity assumption, and consider adding or removing an agent as the initialization of a \emph{new} MDP with a redefined state space, action space, reward function, and dynamics. In this case, the assumption of average full connectivity over time will be violated.} 

\bluethree{In this paper, we consider a probabilistic form of the dropout problem, where the CP knows each agent's probability of dropout \emph{a priori}. In this case of structural change, the objective is to solve for an optimal policy for the \emph{post-dropout} system based on samples collected from \emph{pre-dropout} system. For applications in which policy switching is not feasible, we also wish to find policies that produce good value for both the pre-dropout and the post-dropout systems.}

\bluethree{To give the problem tractability, we assume the following structure to link the post-dropout MDP to its pre-dropout counterpart. First, we assume that the system under no dropout is the nominal graph, from which edges are selected based on agent attendance. Second, each agent has an independent probability of leaving the system, which is known \emph{a priori}. Third, the new transition kernel is assumed to be the corresponding marginalization of the original function. Finally, no reward is associated with removed agents.}


\bluethree{Under these assumptions, the model-based version of the stated objectives are straightforward to solve, and so we focus on model-free solution techniques. Standard RL methods require two basic steps of policy evaluation and policy search. For example, one method for safe policy improvement is to evaluate policies with a high probability confidence bound and select new policies associated with high lower bounds \cite{thomas2015high}. Attempting to perform policy evaluation on the post-dropout system, however, is not straightforward, as we can only sample from the pre-dropout system; the post-dropout system does not yet exist. In addition, the policy search step is difficult, as testing post-dropout policies on the pre-dropout system may result in poor performance.}

\bluethree{We studied a form of this problem in \cite{fiscko2022confident}, in which we proposed a policy evaluation method based on policy importance sampling (IS) \cite{jiang2016doubly} for deterministic dropout of a single agent. Policy IS is desirable for the dropout problem as it successfully avoids the issue of sampling from the pre-dropout system with a bad policy. Policy IS cannot be applied directly to this problem, however, as the pre- and post-dropout systems have different models, meaning the IS estimator cannot be evaluated. The proposed solution enabled IS by leveraging structural connections between the pre-and post-dropout MDPs. We demonstrated that the target policy can be represented in the dimension of the pre-dropout policy; then, standard IS can be applied to pre-dropout trajectories, and marginalization to finds the value post-dropout. }

\subsection{Main Contributions}
\bluethree{In this paper, we expand on the initial results of \cite{fiscko2022confident} to consider probabilistic dropout of many agents. We present a probabilistic dropout problem, in which a policy that maximizes the value of the MDP over its expected agent composition is desired. Note that finding policies for any specific dropout realization is a special case of this problem. }

\bluethree{A naive approach would estimate a policy's value for each system realization and then take the expectation; however, $N$ agents leads to $2^N$ system realizations with values needing estimation. As each estimate needs a dataset of trajectories for IS, this approach quickly becomes untenable.}

\bluethree{In this paper, we show that an alternate approach for solving the probabilistic dropout problem is to define an equivalent single MDP, which we call the \emph{robust MDP}. We develop a policy IS technique to estimate values for the robust MDP. Thus, this formulation may be used to evaluate control policies for post-dropout realizations given samples generated by the pre-dropout MDP. This avoids the complexity issues of the naive approach, and reduces the scale of the problem. }

In Section \ref{sec: prelim}, the multi-agent MDP model is defined. The probabilistic agent dropout problem is defined in Section \ref{sec: problem statement}. The robust MDP formulation is presented in Section \ref{sec:robust}, and the model-free policy IS method and its performance are discussed in Section \ref{sec: policy evaluation}. Finally, simulations are in Section \ref{sec: sims}.

\section{Preliminaries} \label{sec: prelim}
\subsection{Multi-Agent MDPs}
Consider a multi-agent system modeled as a \emph{Markov Decision Process} (MDP) $\mathcal{M}=(\X, \A, r, T, \gamma)$ consisting of the state space of the system $\X$, the action space $\A$, a reward function $r:\X\times\A\to\mathbb{R}$, a probabilistic state-action to state transition function $T:\X\times\A \to\Delta(\X)$, and a discount parameter $\gamma\in(0,1)$.

Next, consider a finite set of agents $\N = \{1,\dots, N\}$. Each agent selects some substate $x_n\in \X_n$ where $x_n$ could model some environmental state and/or personal action. We assume that the sizes $|\X_n|$ are finite and identical across $n$. A state of the MDP is thus the behavior across all the agents $x=\{x_1,\dots, x_N\}$ and the state space is $\X = \bigotimes_{n\in\N}X_n$. In general, a state written with a set subscript such as $x_{\mathcal{B}}$ refers to the actions realized in state $x$ by the agents in $\mathcal{B}$, i.e., $x_{\mathcal{B}} = \{x_b|b\in\mathcal{B}\}$, and the notation $- \tilde{n}$ will refer to the set $\{{n}| {n}\in \N, {n}\neq \tilde{n}\}$.

In this setup, we consider a CP that can broadcast a unique signal $\alpha_n\in\A_n$ to each agent, where $\A_n$ is a finite set of options. The overall control space is thus $\A = \bigotimes_{n\in\N} \A_n$ where one element is $\abold = \{\alpha_1,\dots,\alpha_N\}$. 

Next we consider the \emph{state transition function}, which defines the state-action to state transition densities in the form $p(x'|x,\abold)$. By design, each agent only sees the signal assigned to it from the CP. Furthermore, the agents are connected by a communication structure encoded by a directed graph $(\N, G)$. Then, with $pa(n)$ denoting the parent set of $n$ as given by $G$, we assert the following assumption.

\vspace{1mm}
\begin{assumption} \label{as: agent behavior}
The agents' decision processes are Markovian and time-homogeneous:
\ifarxiv
\begin{equation}
    p(x_n^{t+1}|x^t,\dots, x^0, \alpha_n^t,\dots,\alpha_n^0)=p(x_n^{t+1}|x^t_{pa(n)},\alpha_n^t),\ \forall\  \alpha_n\in \A_{n},\ x_n\in \X_n,\ n\in \N,\ x\in\X,\  t\geq 0.
\end{equation}
\fi
Furthermore, each agent's decision process is independent of the CP actions assigned to the other agents:

\ifarxiv
\begin{equation}
    p(x_n'|x_{pa(n)},\abold)=p(x_n'|x_{pa(n)},\alpha_n),\ \forall\ \alpha_n\in \A_{n},\ x_n\in \X_n,\ n\in \N,\ x\in\X,\  t\geq 0. \label{only signal}
\end{equation}
\fi
\end{assumption}

Equation \ref{factored structure} describes MASs where each agent makes their decision independently after observing the current state and control. For example, this includes general non-cooperative games \cite{osborne1994course}. Furthermore, the time homogeneity property means that the agents have learned their decision processes \emph{a priori}, such as through MARL, game theory, or another paradigm. The CP is agnostic to the learning processes used by the agents as long as they satisfy the Markov and time homogeneity assumptions. 

As a result of Assumption \ref{as: agent behavior}, the overall state-action to state transition probabilities satisfy the following factored structure:
\begin{align}
    p(x'|x,\abold) = \prod_{n\in\N} p(x_n'|x_{pa(n)},\alpha_n).\label{factored structure}
\end{align}
MDPs whose transitions satisfy \eqref{factored structure} are known as \emph{transition independent MDPs} (TI-MDPs). For ease of notation, the explicit dependence on $pa(n)$ may be dropped, which implicitly defines the most general model of a fully-connected graph.  Next, the following assumption will be made on the reward function of the MDP.

\begin{assumption}\label{as:reward}
The reward function $r(x,\abold)$ satisfies the following \emph{separable} structure:
\begin{align}
    r(x,\abold) \triangleq \sum_{n\in\N}r_n(x_n,\alpha_n),
\end{align}
where each function $r_n$ is non-negative, deterministic, and bounded for all $x\in\X$, $\abold\in\A$.
\end{assumption}

Assumption \ref{as:reward} means that the CP's objective can be encoded per individual agent. For example, this reward can be the proportion of agents in desired goal states. TI-MDPs with separable reward functions defined over the same scope are known as \emph{factored MDPs} \cite{guestrin2003efficient}, \cite{osband2014near}. Through the rest of the paper, it will be assumed that Assumptions \ref{as: agent behavior} and \ref{as:reward} hold.

Finally, the CP may solve the MDP for some policy  $\pi:\X\to\A$. In this work we consider the standard discounted \emph{value function}. The finite horizon value function is,

\small
\ifarxiv
\begin{equation}
    V^{\pi}_H(x) = \mathbb{E}\left[\sum_{t=0}^{H}\gamma^t r(x_t,\abold_t)\vert x_0=x, \abold_t\sim \pi(x_t), x_{t+1}|x_t,\abold_t\sim T \right].
\end{equation}
\fi
\normalsize

The infinite horizon value function can similarly be calculated as, $V^{\pi}(x) = \lim_{H\to\infty}V^{\pi}_H(x).$ For brevity, the notation $V^{\pi}\in\R^{|\X|}$ will refer to the vector of values $V^{\pi}(x)\ \forall\ x\in\X$. An optimal policy $\pi^*$ is one that maximizes the value function, $\pi^*(x) \in \argmax_{\pi} V^{\pi}(x)$. The optimal value $V^*(x)$ is known to be unique, and for finite stationary MPDs there exists an optimal stationary deterministic policy \cite{bertsekas1995dynamic}. 

Under a fixed policy, the MDP will evolve as a Markov chain. The \emph{stationary distribution} $\mu$ of a Markov chain under policy $\pi$ is the left eigenvector corresponding to eigenvalue 1 of the transition matrix.

Commonly used in dynamic programming methods to solve an MDP is the Bellman operator. 
\vspace{1mm}
\begin{definition}
The \emph{Bellman operator} applied to the value function $V(x)$ is defined as,
\ifarxiv
\begin{align}
    \Tb^{\pi} V(x) &=  \mathbb{E}_{\abold}\left[ r(x,\abold)+\gamma V(x') \Big\vert \abold\sim\pi,\ x'\vert x,\abold\sim T,\ x_0 = x \right]\label{bellman},\\
    \Tb V(x) &=  \max_{\abold\in\A}\mathbb{E}\left[ r(x,\abold)+\gamma V(x') \Big\vert \ x'\vert x,\abold\sim T,\ x_0 = x \right].
\end{align}
\fi
\end{definition}

\section{Problem Statement} \label{sec: problem statement}
Consider a multi-agent system modeled as the MDP $\mm$. While the CP may have found a policy $\pi$ that produces high value, the goodness of this policy is only valid for as long as the model $\mm$ is stationary. The objective in this section is to formalize the model of the post-dropout system and discuss its structural relationship to the pre-dropout system.

\subsection{Probabilistic Agent Dropout}
Let the probability of dropout for agent $n$ be $1-\beta_n$, which is sampled independently of the other agents. Let the vector of probabilities be denoted by $B=[\beta_1,\dots,\beta_N]$. \bluethree{Let the vector $W$ denote a realization of the system via binary flags where $w_n = 1$ means the agent is in the system (i.e., active) and $w_n = 0$ means the agent has left the system (i.e., inactive). Define the shorthand $\N_W = \{n|w_n = 1\}$ to be the set of agents active within the system, and similarly the shorthand $-\N_{W} = \{n|w_n = 0\}$ to be the set of dropped agents. Note that fixing a $W$ is equivalent to sampling a graph with edge weights associated with $n\in -\N_W$ set to zero and $n\in\N_W$ set to one. Finally, let the \emph{expected system} be the graph resulting from taking the expectation of the edge weights across all $W$. The MDP parameters must be adjusted to account for these new graph connections, as shown in Figure \ref{fig: system ex}.}

\begin{figure}
    \centering
    \includegraphics[width=1\linewidth]{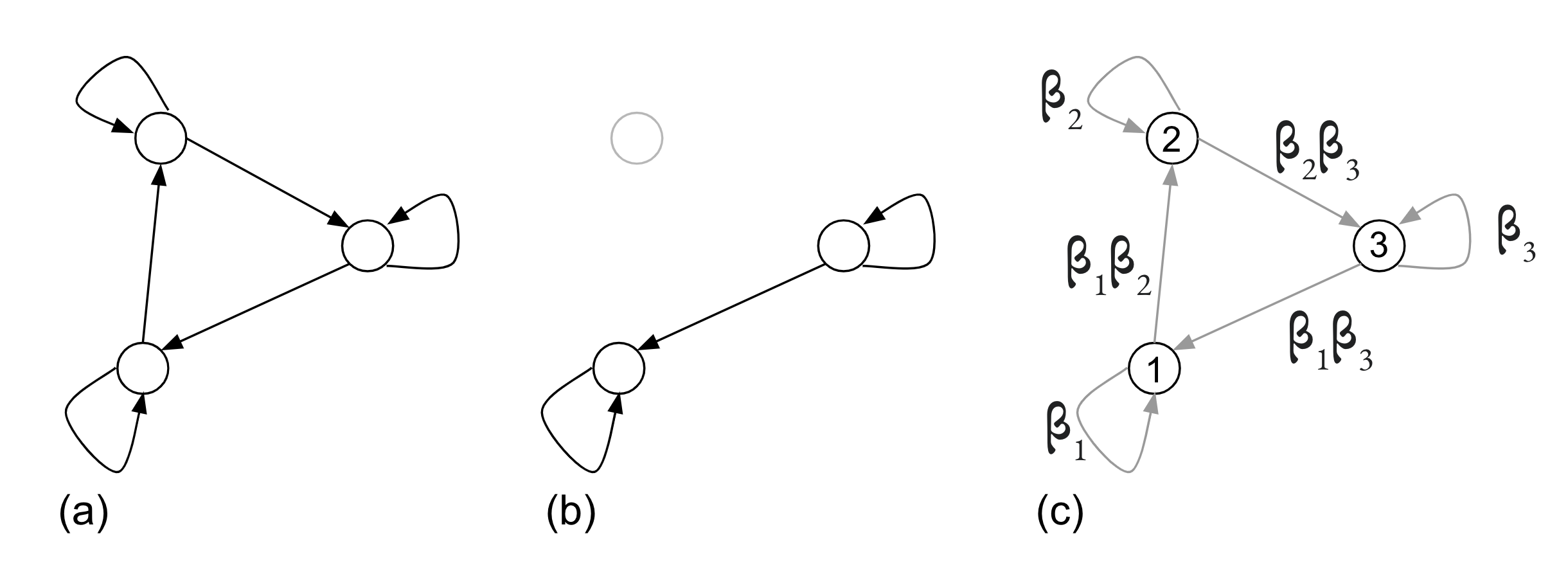}
    \caption{\bluethree{Example system graph. Part (a) shows system under no dropout, i.e., the nominal graph. Part (b) shows a system dropout realization. Part (c) shows the expected system, where the weight of each link is scaled by the corresponding $\beta$.}}
    \label{fig: system ex}
\end{figure}

    \bluetoo{A \emph{realization} of the multi-agent MDP subject to dropout is denoted as} $(\mathcal{M}|W) = (\N, \X, \A, T, r, \gamma| W)$, where:
    \begin{itemize}
        \item $\{\X|W\} \triangleq\bigotimes_{n\in\N_W}X_n$, where one state given the dropout configuration is referred to as $\bar{x}$.
        \item $\{\A|W\} \triangleq \bigotimes_{n\in\N_W}\A_n$, where one action given the dropout configuration is referred to as $\bar{\abold}$.
        \item $T(x,\abold, x'|W): \{\X|W\}\times\{\A|W\}\to \Delta \{\X|W\}$.
        \item $r(x,\abold|W)\triangleq \sum_{n\in\N_W} r_n(x_n,\alpha_n|w_n = 1)$, where the rewards to dropped agents are defined to be zero: $r_n(x_n,\alpha_n|w_n = 0)\triangleq 0$. 
    \end{itemize}
    The policy and value of the realized MDP are:
    \begin{itemize}
        \item $\pibar(\abar|\xbar): \{\X|W\}\to \Delta\{\A|W\}$, where the policy conditioned on the dropout configuration is $\pibar = \prod_{n\in\N_W}P(\abold_n|\xbar)$.
        \item \bluethree{$V^{\bar{\pi}}(\bar{x}|W)$ with equivalent shorthand $V^{\pi}(x|W)$} and corresponding finite time value $V^{\pibar}_H(\xbar|W)$. 
    \end{itemize}
\bluethree{The policy for the expected system is defined as $V_R^{\pi}(x)\triangleq \E_W[V^{\pi}(x|W)]$.}

\bluetoo{The \emph{pre-dropout} system is equal to $(\mm|W = \mathbf{1})$. This is equivalent to the original MDP $\mm$.  }

\bluethree{Given $\mu(x)$, the joint distribution  $\mu(x_{-\N_W})$ of the removed agents is the summation, 
\begin{align}
    \mu(x_{-\N_W}) &=\sum_{n\in \N_W} \mu(x_1,\dots,x_n,\dots,x_N). \label{mu big}
\end{align}
}

\begin{assumption} \label{ergodic ass}
    The pre-dropout system is an ergodic Markov chain under any fixed policy.
\end{assumption}
\begin{assumption}
    \bluethree{The state-action-state probabilities of system $(\mm|W)$ marginalize the inactive agents $\{n|w_n = 0\}$ from the transition probabilities of the nominal model $\mm$:}
    \begin{align}
        P(x_n'|\xbar,\abar_n)&=\E_{x_{-\N_{W}}}\E_{\alpha_{-\N_W}}[P(x_n'|x,\abold)], \label{agent marg}
    \end{align}
    \bluethree{where the marginal probability of $\alpha$ comes from $\pi$ and $x_{-\N_{W}}$ comes from \eqref{mu big}.}
\end{assumption}

The overall transitions may be related as,
\begin{align}
    P(\xbar'|\xbar,\abar) & = P(\xbar'|x_{\N_W},\alpha_{\N_W}),\nonumber\\
    &=\E_{x_{-\N_W},\alpha_{-\N_W}}P(\xbar'|x_{\N_W},x_{-\N_W},\alpha_{\N_W},\alpha_{-\N_W}),\nonumber\\
    &=\E_{x_{-\N_W},\alpha_{-\N_W}}P(\xbar'|x,\abold). \label{marg t}
\end{align}

\bluetoo{The main goals in this paper can thus be summarized as the following objectives:}

\bluethree{\textbf{Problem 1:} Evaluate policy performance for specific system realizations: find $\hat{V}^{\pi}(x|W)\approx V^{\pi}(x|W)$.} \label{prob: 1}

\bluethree{\textbf{Problem 2:} Evaluate policy performance for the expected system: find $\hat{V}_R^{\pi}(x)\approx V_R^{\pi}(x)$.} \label{prob: 2}

\bluethree{The quantity $V_R^{\pi}(x)$ is called the \emph{robust value} as a policy that performs well by this metric provides good value for the expected system. }

If the model $T$ of the original system is known, then the optimal post-dropout policy may \bluetoo{be} easily computed. We thus focus our attention on the case where the model $T$ is unknown and a policy must be learned from experience. All trajectories are sampled from the \emph{original system}, i.e., the multi-agent MDP before any dropout has occurred. This has two critical implications. First, we are unable to sample from any realization of the post-dropout MDP as these systems do not yet exist. Second, we cannot exert a ``bad" policy on the existing system; we always want to control the existing system with a policy that yields an acceptable amount of value. These observations are formalized in the following assumption. 

\begin{assumption}\label{known}
The CP can generate trajectories from the pre-dropout system $\mm$ and cannot sample from any post-dropout realization of the system. The generating transition function $P(x'|x,\abold)$ is unknown. 
\end{assumption}

\subsection{Policy Importance Sampling}
\bluetoo{To safely control the existing system while simultaneously enabling exploration of the policy space, we are interested in a policy IS method. 
Given a dataset of trajectories generated by some policy $\pi$, the goal of policy IS is to estimate the dataset's sample return had the trajectories instead been produced by a different policy $\phi$.
In this way, the system can be controlled with a policy that produces good value while facilitating evaluation of alternate policies. Next, we recap standard policy IS and introduce prior work on dropout evaluation. }

The sample return of a trajectory with horizon length $H$ beginning at $x_0 = x$ is defined as,
\begin{equation}
    G_H(x) \triangleq \sum_{t=0}^{H-1}\gamma^{t}r(x_t,\abold_t).
\end{equation}

Given trajectory $\tau=(x_0,\abold_0,r_0,\dots, x_{H-1},\abold_{H-1}, r_{H-1})$ generated by behavioral policy $\pi$, policy IS estimates the return had $\tau$ instead been generated by target policy $\phi$. For tractability, the following assumption must hold. 

\begin{assumption}\label{support}
$\phi$ is fully supported on $\pi$, i.e. for $\abold$ such that $\phi(\abold|x)>0$ then $\pi(\abold|x)>0$.
\end{assumption}

With $p$ as the joint distribution of $\tau$ under $\phi$ and $q$ the joint distribution of $\tau$ under $\pi$, the estimated \bluethree{return of $\tau$ under $\phi$} is,
\ifarxiv
\begin{equation}
     V_H^{\phi}(x) = \E_{\tau\sim q}\left[\frac{p(\tau)}{q(\tau)}\sum_{t=0}^{H-1} \gamma^{t} r_t(x_t,\abold_t) \right] = \E_{\tau\sim q}\left[\frac{p(\tau)}{q(\tau)}G_H(x) \right].\label{IS original}
\end{equation}
\fi
In comparing two policies on the same MDP, the IS ratio is,
\ifarxiv
\begin{equation*}
    \frac{p(\tau)}{q(\tau)} = \frac{d(x_0)\phi(\abold_1|x_1)P(x_2|x_1,\abold_1)\dots \phi(\abold_H|x_H)}{d(x_0)\pi(\abold_1|x_1)P(x_2|x_1,\abold_1)\dots \pi(\abold_H|x_H)} = \frac{\phi(\abold_1|x_1)\dots\phi(\abold_H|x_H)}{\pi(\abold_1|x_1)\dots\pi(\abold_H|x_H)} = \prod_{t=1}^H\frac{\phi(\abold_t|x_t)}{\pi(\abold_t|x_t)},
\end{equation*}
\fi
which only depends on the chosen policies. The return estimate \eqref{IS original} can thus be evaluated purely from sampled and known information.

In the dropout scenario, however, the objective is to evaluate policies for the \emph{post-dropout} system. Evaluating the probability ratio for a dropout realization $(\mm|W)$ yields,
\begin{equation}
     \frac{p(\tau)}{q(\tau)} = \frac{d(\xbar_0)\phi(\abar_1|\xbar_1)P(\xbar_2|\xbar_1,\abar_1)\dots \phi(\abar_H|\xbar_H)}{d(x_0)\pi(\abold_1|x_1)P(x_2|x_1,\abold_1)\dots \pi(\abold_H|x_H)},\label{ratiobad}
\end{equation}
which will not cancel to a ratio of the policies. This problem arises because the samples were generated from the pre-dropout MDP which has a different model from target post-dropout system.  As \eqref{ratiobad} cannot be evaluated, policy IS for a specific dropout realization may not be used by directly applying the existing method. 

In \cite{fiscko2022confident}, we investigated adapting policy IS to the dropout scenario under the case of deterministic dropout of one agent. It is straightforward to expand this approach to consider multiple dropped agents and perform policy evaluation for one $(\mm|W)$. One strategy to find an acceptable post-dropout policy could be to implement a  policy search routine that leverages this policy IS method for the policy evaluation step. While this proposed method is valid, it scales poorly in terms of the data required to find acceptable policies for all possible combinations of dropped agents. For a set of $N$ agents, there are $2^N$ possible realizations of $W$; therefore to fully solve the post-dropout system, the solution algorithm must be run $2^N$ times. In terms of computation and data complexity, this proposed approach becomes untenable. In the following section, we will demonstrate an alternate method.

\section{Robust MDP} \label{sec:robust}
In this section, we develop an analytical understanding of node dropout in the multi-agent MDP. In particular, we demonstrate that the probabilistic dropout problem can be reduced to a single MDP, dubbed the \emph{robust MDP}, which can be used to find a control policy robust to dropout. 

\blue{The next definition establishes the robust MDP, which considers the expected system over all dropout realizations. The set of agents, states, and actions are identical to the original system, as is the transition function. The key difference, however, comes from the reward function, which is defined now as the expected reward across all dropout realizations. }

\begin{definition}\label{robust MDP}
    Define the \textbf{robust multi-agent MDP} $\mathcal{M}^R = (\N, B, \X, \A, T, r^R, \gamma)$, where the reward is defined as:
    \begin{align}
        r^R(x,\abold) &= \sum_{n\in\N}r_n^R(x_n,\alpha_n)\\
        r_n^R(x_n,\alpha_n) &\triangleq \E_{w_n}[r_n(x_n,\alpha_n| w_n)]=\beta_n r_n(x_n,\alpha_n|w_n=1)
    \end{align}
\end{definition}
\blue{As a result of the new reward function, we must specify the value function. Let $J^{\pi}(x) = \lim_{H\to\infty}\E_{\pi}\Big[\sum_{t=0}^{H-1}\gamma^tr^R(x_t,\abold_t)\Big]$ be the infinite-horizon value function of this system $\mathcal{M}^R$. The function $J$ satisfies,}
\begin{align}
    &J^{\pi}(x) = \E_{\abold}\left[\E_W[r(x,\abold|W)]+\gamma \E_{x'}[J^{\pi}(x')|x,\abold]\right].\label{J def}
\end{align}
As $B$ is known \emph{a priori}, the expected rewards over dropout realizations can easily be evaluated. This fact, combined with the assumption that trajectories from $T$ can be generated, means that $J$ can be estimated in a model-free setting.

\bluethree{An important distinction at this point is that the quantity $J^{\pi}(x)$ is not necessarily equal to the desired robust value $V_R^{\pi}(x)$. The next section will establish the mathematical relationship between these functions.}

\subsection{Value Theorems}
In this section, an analysis of the robust MDP will be performed to provide structure for the probabilistic dropout problem. \bluethree{The first theorem expresses the value of a post-dropout value function with respect to a given pre-dropout policy. At this stage, we seek to reconcile the difference in dimension between the pre-dropout and post-dropout systems, and evaluating different policies will be achieved later by IS.}

\begin{theorem} \label{marginalization theorem}
    \textbf{(Value of a Realized Post-Dropout System)} Consider a realization of the system $W$ and some policy $\pibar$. Define the policy,
    \begin{align}
        \pi(\abold|x) = \pibar(\abar|\xbar)\prod_{n\in-\N_W}\bluethree{\pi_n(\abold_n|x)}.\label{agument}
    \end{align}

    \noindent Then the finite horizon value of system $(\mm|W)$ is equal to, 
    \ifarxiv
    \begin{equation}
        \Vbar_H^{\pibar}(\xbar)= \E_{x_{-\N_W}}\E\left[\sum_{t=0}^{H-1}\gamma^tr(x_t,\abold_t|W)\big|x_0 = x, \abold_t|x_t\sim \pi,  x_{t+1}|x_t, \abold_t\sim T\right], \label{conversion theorem eq} 
    \end{equation}
    \fi
    and the infinite horizon value is equal to,
    \ifarxiv
    \begin{equation}
        \Vbar^{\pibar}(\xbar)=\lim_{H\to\infty}\E_{x_{-\N_W}}\E\left[\sum_{t=0}^{H-1}\gamma^tr(x_t,\abold_t|W)\big|x_0 = x, \abold_t|x_t\sim \pi,  x_{t+1}|x_t, \abold_t\sim T\right], 
    \end{equation}
    \fi
    where the expectation is with respect to the marginal distribution of the pre-dropout system under $\pi$.
\end{theorem}

As a consequence of Theorem \ref{marginalization theorem} and the policy design in \eqref{agument}, it holds that $\Vbar^{\pibar}(\xbar) = V^{\pi}(x|W)$.

\begin{proof}
    In this proof we construct a fictitious MDP with state space $\X$, action space $\A$, transition $T$, and reward function $r(x,\abold|W)$. This system is identically distributed to the pre-dropout system. We will use the notation $U$ to denote the value function associated with this system. To prove the theorem, the goal will be to show that the expression $\Vbar^{\pibar}(\xbar) = \E_{x_{-\N_W}}U_t^{\pi}(x)$ holds for all $t$. Initial case:
    \begin{align}
        &\Vbar_{H}^{\pibar}(\xbar)=\E_{\abar}[r(\xbar,\abar)]= \E_{\abar}\left[\sum_{n\in\N_W}r_n(x_{n},\alpha_n) \right]. \label{eq: a init}
    \end{align}
    Note that for $n\in\N_W$, $r_n(x_{n},\alpha_{n}|w_n = 1)$ is independent of $x_{-\N_W}$ and $\abold_{-\N_W}$ because agents in $-\N_W$ have been removed from the system. Furthermore, $r_n(x_{n},\alpha_{n}|w_n = 0) = 0$. Let $\pi_U(\abold_{-\N_W}|x)=\prod_{n\in-\N_W}\pi(\abold_n|x)$. Then for $x_{-\N_W}$ and $\abold_{-\N_W}$ sampled from the fictitious system, \eqref{eq: a init} is:
    
\ifarxiv

\begin{align*}
        &\sum_{x_{-\N_W}}\mu(x_{-\N_W})\sum_{\abold_{-\N_W}}\pi_U(\abold_{-\N_W}|x)\sum_{\abar}\pibar(\abar|\xbar)\bigg[\sum_{n\in\N_W}r_n(x_{n},\alpha_n) \bigg],\\
    &= \sum_{x_{-\N_W}}\mu(x_{-\N_W})\sum_{\abold}\pi(\abold|x)\left[\sum_{n\in\N_W}r_n(x_{n},\alpha_n) \right],\\
    &=\E_{x_{-\N_W}}\E_{\abold}\bigg[\sum_{n\in\N_W}r_n(x_{n},\alpha_{n}|w_n = 1)+\sum_{n\notin\N_W}r_n(x_{n},\alpha_{n}|w_n = 0)\bigg],\\
    &=\E_{x_{-\N_W}}\E_{\abold} [U_H^{\pi}(x)].
\end{align*}
\fi

\bluetoo{Now} use induction to show for general $t$.  Note that as $x_{-\N_W}$ and $\abold_{-\N_W}$ are sampled from the fictitious system, the transition matrices are related according to \eqref{marg t}. Assume the claim holds for $t+1$: $\Vbar^{\pibar}_{t+1}(\xbar') = \E_{x_{-\N_W}'}U_{t+1}^{\pi}(x')$. Then stepping backwards,
\ifarxiv
\begin{align}
    \Vbar_{t}^{\pibar}(\xbar)&=\E_{\abold_{\N_W}}\left[r(x,\abold|W)+\gamma\E_{\xbar'} [\Vbar^{\pibar}_{t+1}(\xbar')|\xbar,\abar]\right]\nonumber \\
    &=\E_{\abold_{\N_W}}[\E_{x_{-\N_W},\alpha_{-\N_W}}r(x,\abold|W)+\gamma\E_{x_{-\N_W},\alpha_{-\N_W}}\E_{\xbar'} [\E_{x_{-\N_W}'}U_{t+1}^{\pi}(x')|x,\abold]]\label{iterated exp}\\
    &=\E_{x_{-\N_W}}\E_{\abold}[r(x,\abold|W)+\gamma \E_{x'} [U_{t+1}^{\pi}(x')|x,\abold]]\label{iterated exp next}\\
    &=\E_{x_{-\N_W}}[U_{t}^{\pi}(x)]\nonumber.
\end{align}
\fi

As the pre-dropout system is assumed to be an ergodic chain for all policies in Assumption \ref{ergodic ass}, the marginalization will be well defined. To verify the infinite horizon case, note that all rewards are non-negative and 
\ifarxiv
\begin{equation}
    \Vbar^{\pibar}(\xbar)=\E_{x_{-\N_W}}\E\left[\sum_{t=0}^{H-1}\gamma^tr(x_t,\abold_t|W)\right]+\E_{x_{-\N_W}}\E\left[\sum_{t=H}^{\infty}\gamma^tr(x_t,\abold_t|W)\right].
\end{equation}
\fi
The right most term is lower bounded by 0 as $r$ is  non-negative. In addition, it can be expressed as
\ifarxiv
\begin{align*}
    &\E_{x_{-\N_W}}\E\left[\sum_{t=H}^{\infty}\gamma^tr(x_t,\abold_t|W)\right]\leq \E_{x_{-\N_W}}\rmax\frac{\gamma^H}{1-\gamma} = \rmax\frac{\gamma^H}{1-\gamma}.
\end{align*}
\fi
Take the limit as $H\to\infty$ to show the upper bound goes to 0 to complete the proof.
\end{proof}

Theorem \ref{marginalization theorem} establishes that samples generated by the pre-dropout MDP can be used to evaluate policies for the post-dropout system based on a marginalization relationship. This resolves the issue of being unable to sample from the desired post-dropout realization. Building on this observation, the next result states that the value function of the robust MDP is equal to the value of the expected realization of the system.

\begin{theorem} \label{dropout value function theorem}
    \textbf{(Value of the Expected System)} The value of the expected system is, 
    \begin{align}
        V^{\pi}_R(x) &= \E_W\E_{x_{-\N_W}} J^{\pi}(x),\label{value dropout}%
    \end{align}
    and a robust policy can be computed as,
    \begin{align}
        \pi_R(\abold|x)  =  \E_W\E_{x_{-\N_W}}\pi(\abold|x). \label{policy dropout}
    \end{align}
\end{theorem}
\begin{proof}
    Show that the value function of the expected system is equivalent to the robustness criterion. 
    
      Initial case. Note that again $r_n=0$ for removed agents in $-\N_W$. Thus the expectation can be taken, $V_{RH}^{\pi}(x)=\E_{\abold}\E_W[r(x,\abold|W)]= \E_W\E_{x_{-\N_W}}\E_{\abold}\E_W[r(x,\abold|W)]$.

    Next use induction to show for general $t$. Assume the claim holds for $t+1$: $V^{\pi}_{R,t+1}(x') = \E_{x'_{-\N_W}}J^{\pi}_{t+1}(x')$. Then stepping backwards,%
    \begin{align}
    &V_{R,t}^{\pi}(x) = \E_W [V^{\pi}_t(x|W)]\nonumber\\
    &= \E_W\E_{\abar}[r(x,\abold|W)+\gamma \E_{\xbar'}[V^{\pi}_t(x')|\xbar,\abar]]\nonumber\\
    &= \E_W\E_{\abar}[r(x,\abold|W)+\gamma\E_{\xbar'}[\E_{x'_{-\N_W}}J^{\pi}_{t+1}(x')|\xbar,\abar]]\label{iterated exp}\\
    &= \E_W\E_{\abar}[ r(x,\abold|W)+\gamma \E_{x'}[J^{\pi}_{t+1}(x')|\xbar,\abar]]\label{iterated exp next}\\
    \begin{split}
        &= \E_W\E_{\abar}[\E_{x_{-\N_W},\abold_{-\N_W}}r(x,\abold|W)\\
        &\quad+\gamma \E_{x_{-\N_W},\abold_{-\N_W}}\E_{x'}[J^{\pi}_{t+1}(x')|x,\abold]]\nonumber
    \end{split}\\
    &=\E_W\E_{x_{-\N_W}}\E_{\abold}[r(x,\abold|W)+\gamma \E_{x'}[J^{\pi}_{t+1}(x')|x,\abold]]\label{vrt last}
    \end{align}
    Then note that as no reward is given to removed agents:
    \begin{align}
        &\E_W\E_{x_{-\N_W}}\E_{\abold}[r(x,\abold|W)]\nonumber\\
        &=\E_W\E_{\abold}[r(x,\abold|W)]\nonumber\\
        &=\sum_{n\in\N}\E_{w_n}\E_{\abold_n}[r(x_n,\abold_n|w_n)]\nonumber\\
        &=\sum_{n\in\N}\E_{\abold_n}\E_{w_n}[r(x_n,\abold_n|w_n)]\nonumber\\
        &=\E_{\abold}[\E_Wr(x,\abold|W)]\nonumber\\
        &=\E_W\E_{x_{-\N_W}}\E_{\abold}[\E_Wr(x,\abold|W)]\label{lots of e}
    \end{align}
    Then continuing from \eqref{vrt last},
    \begin{align}
        &=\E_W\E_{x_{-\N_W}}\E_{\abold}[r(x,\abold|W)+\gamma \E_{x'}[J^{\pi}_{t+1}(x')|x,\abold]]\nonumber\\
        &=\E_W\E_{x_{-\N_W}}\E_{\abold}[\E_Wr(x,\abold|W) +\gamma \E_{x'}[J^{\pi}_{t+1}(x')|x,\abold]]\nonumber\\
        &=\E_W\E_{x_{-\N_W}} J^{\pi}_t(x)\label{exp j 2}
    \end{align}

    To verify the infinite horizon case, note that all rewards are non-negative and, 
    \ifarxiv
    \begin{equation*}
        V^{\pi}_R(x)=\E_W\E_{x_{-\N_W}}\E\left[\sum_{t=0}^{H-1}\gamma^t\E_Wr(x_t,\abold_t|W)\right]+\E_W\E_{x_{-\N_W}}\E\left[\sum_{t=H}^{\infty}\gamma^t\E_Wr(x_t,\abold_t|W)\right].
    \end{equation*}
    \fi
    The \bluetoo{second} term is lower bounded by 0 as $r$ is bounded to be non-negative. In addition, it can be expressed as,
    \ifarxiv
    \begin{align*}
        &\E_W\E_{x_{-\N_W}}\E\left[\sum_{t=H}^{\infty}\gamma^t\E_Wr(x_t,\abold_t|W)\right]\leq \E_{x_{-\N_W}}\rmax\frac{\gamma^H}{1-\gamma} = \rmax\frac{\gamma^H}{1-\gamma}.
    \end{align*}
    \fi
    Take the limit as $H\to\infty$ to show the upper bound goes to 0 to complete the proof. Equation \eqref{policy dropout} then follows.
    \end{proof}

\begin{remark}
    Retaining the transition function as the original $T$ is a key necessary condition for policy IS. Replacing the reward with $\E_W r(x,\abold|W)$ will then enable policy IS to be evaluated on a single system, rather than needing estimate a value for each of the $2^N$ combinations of $W$.
\end{remark}

Given Theorem \ref{dropout value function theorem}, it can be established that the optimal robust value and policy can similarly be found as \bluetoo{an expectation}.

\begin{theorem}
    \bluetoo{\textbf{(Optimal Value and Policy of the Expected System)}} The optimal robust value satisfies, 
    \begin{equation}
        V^{*}_R(x) = \bluethree{\max_{\pi}V_R^{\pi}(x)} = \E_{W\sim B}[V^{*}(x|W)],\label{opt robust value}
    \end{equation}
    and an optimal robust policy \bluethree{$\pi^*_R(\abold|x)\in\argmax_{\pi}V_R^{\pi}(x)$} satisfies,
    \begin{equation}
        \pi^{*}_R(\abold|x) = \E_{W\sim B}[\pi^{*}(\abold|x,W)].\label{opt robust policy}
    \end{equation}
\end{theorem}
\begin{proof}
\bluetoo{This follows} from Theorem \ref{dropout value function theorem}.
\end{proof}

\subsection{\bluetoo{Performance of the Robust Policy}}
The next two results will relate the performance of a single policy across both the original system $\mm$ and the robust model $\mm_R$. Given the value of a policy on the pre-dropout system, this first result evaluates the robust policy's performance.

\begin{lemma} \label{other vr representation}
    Let all agents have identical probabilities $\beta_n \equiv \beta$. The performance of some policy $\pi$ on the robust model $\mm_R$ and the original system $\mm$ can be related as,
    \begin{align}
        V^{\pi}_R(x) &= \beta \E_W\E_{x_{-\N_W}} V^{\pi}(x|W = \mathbf{1}).
    \end{align}
\end{lemma}

\begin{proof}
    \ifarxiv
    Consider the value function $J$ as in \eqref{J def} with the expected reward function and original transitions. 
    Base case of induction:
    \ifarxiv
    \begin{align*}
        J^{\pi}_0(x) &=\E_{\abold}[r^R(x,\alpha)],\\
        &=\E_{\abold}\E_W[r(x,\alpha|W)]\\
        &=\E_{\abold}\sum_{n\in\N}\E_{w_n}[r_n(x_n,\alpha_n|w_n)],\\
        &= \E_{\abold}\sum_{n\in\N}\left(\beta[r_n(x_n,\alpha_n|w_n = 1)]+(1-\beta)[r_n(x_n,\alpha_n|w_n = 0)]\right),\\
        &=\beta\E_{\abold}[r(x,\abold|W = \mathbf{1})]+(1-\beta)\E_{\abold}[r(x,\abold|W= \mathbf{0})].
    \end{align*}
    \else
    \begin{align*}
        J^{\pi}_0(x) &=\E_{\abold}[r^R(x,\alpha)],\\
        &=\E_{\abold}\E_W[r(x,\alpha|W)],\\
        &=\E_{\abold}\sum_{n\in\N}\E_{w_n}[r_n(x_n,\alpha_n|w_n)],\\
        \begin{split}
            &= \E_{\abold}\sum_{n\in\N}\left(\beta[r_n(x_n,\alpha_n|w_n = 1)]\\
            &\quad+ (1-\beta)[r_n(x_n,\alpha_n|w_n = 0)]\right),
        \end{split}\\
        &=\beta\E_{\abold}[r(x,\abold|W = \mathbf{1})]+(1-\beta)\E_{\abold}[r(x,\abold|W= \mathbf{0})].
    \end{align*}
    \fi
    Then:
    \ifarxiv
    \begin{align}
        J^{\pi}_1(x) &= \E_{\abold}[r^R(x,\abold)+\gamma \E_{x'}[J_0^{\pi}(x')|x,\abold]],\nonumber\\
        &=\E_{\abold}[\E_{W}[r(x,\abold|W)]+\gamma \E_{x'}[J_0^{\pi}(x')|x,\abold]],\nonumber\\
         &=\E_{\abold}\left[\sum_{n\in\N}\E_{w_n}[r_n(x_n,\alpha_n|w_n)]+\gamma \E_{x'}[\beta\E_{\abold}[r(x',\abold|W = \mathbf{1})]+(1-\beta)\E_{\abold}[r(x',\abold|W= \mathbf{0})]|x,\abold]\right],\nonumber\\
        \begin{split}
            &=\E_{\abold}[\beta r(x,\abold,|W = \mathbf{1}) + (1-\beta) r(x,\abold,|W = \mathbf{0})\\
            &\quad +\gamma \E_{x'}[\beta\E_{\abold}[r(x',\abold|W = \mathbf{1})]+(1-\beta)\E_{\abold}[r(x',\abold|W= \mathbf{0})]|x,\abold]],\nonumber\\
        \end{split}\\
        \begin{split}
             &=\E_{\abold}[\beta r(x,\abold,|W = \mathbf{1}) +\gamma \E_{x'}[\beta\E_{\abold'}[r(x',\abold'|W = \mathbf{1})]|x,\abold]] \\
             &\quad + \E_{\abold}[(1-\beta) r(x,\abold,|W = \mathbf{0})+\gamma \E_{x'}[(1-\beta)\E_{\abold'}[r(x',\abold'|W= \mathbf{0})]|x,\abold]],\nonumber\\
         \end{split}\\
        &= \beta \Tb^{\pi}J_0(x|W=\mathbf{1}) + (1-\beta)\Tb^{\pi}J_0(x|W=\mathbf{0}).\label{induction 0}
    \end{align}
    \else
    \begin{align}
        &J^{\pi}_1(x) \nonumber\\
        &= \E_{\abold}[r^R(x,\abold)+\gamma \E_{x'}[J_0^{\pi}(x')|x,\abold]],\nonumber\\
        &=\E_{\abold}[\E_{W}[r(x,\abold|W)]+\gamma \E_{x'}[J_0^{\pi}(x')|x,\abold]],\nonumber\\
        \begin{split}
            &=\E_{\abold}\Bigg[\sum_{n\in\N}\E_{w_n}[r_n(x_n,\alpha_n|w_n)]\\
            &\quad+\gamma \E_{x'}[\beta\E_{\abold}[r(x',\abold|W = \mathbf{1})]\\
            &\quad+(1-\beta)\E_{\abold}[r(x',\abold|W= \mathbf{0})]|x,\abold]\Bigg],\nonumber
        \end{split}\\
        \begin{split}
            &=\E_{\abold}[\beta r(x,\abold,|W = \mathbf{1})\\
            &\quad + (1-\beta) r(x,\abold,|W = \mathbf{0})\\
            &\quad +\gamma \E_{x'}[\beta\E_{\abold}[r(x',\abold|W = \mathbf{1})]\\
            &\quad +(1-\beta)\E_{\abold}[r(x',\abold|W= \mathbf{0})]|x,\abold]],\nonumber\\
        \end{split}\\
        \begin{split}
             &=\E_{\abold}[\beta r(x,\abold,|W = \mathbf{1}) \\
             &\quad +\gamma \E_{x'}[\beta\E_{\abold'}[r(x',\abold'|W = \mathbf{1})]|x,\abold]] \\
             &\quad + \E_{\abold}[(1-\beta) r(x,\abold,|W = \mathbf{0})\\
             &\quad +\gamma \E_{x'}[(1-\beta)\E_{\abold'}[r(x',\abold'|W= \mathbf{0})]|x,\abold]]\nonumber,\\
         \end{split}\\
        &= \beta \Tb^{\pi}J_0(x|W=\mathbf{1}) + (1-\beta)\Tb^{\pi}J_0(x|W=\mathbf{0}).\label{induction 0}
    \end{align}
    \fi

    Next, the inductive step must be established. Assume the following expression holds for iteration $k$:
    \begin{align*}
        J^{\pi}_k(x) &=\beta \Tb^{\pi}J_{k-1}(x|W=\mathbf{1}) + (1-\beta)\Tb^{\pi}J_{k-1}(x|W=\mathbf{0}),\\
        \begin{split}
            &=\beta \E_{\abold}[r(x,\abold|W = \mathbf{1}) + \gamma\E_{x'}[J_{k-1}^{\pi}(x'|W = \mathbf{1})|x,\abold)]] \\
            &\quad (1-\beta) \E_{\abold}[r(x,\abold|W = \mathbf{0}) + \gamma\E_{x'}[J_{k-1}^{\pi}(x'|W = \mathbf{0})|x,\abold)]] .
        \end{split}
    \end{align*}
    Then for $k+1$:
    \ifarxiv
    \begin{align}
        J_{k+1}^{\pi}(x) &= \E_{\abold}[r^R(x,\abold)+\gamma \E_{x'}[J_k^{\pi}(x')|x,\abold]],\nonumber\\
        &=\E_{\abold}[\E_{W}[r(x,\abold|W)]+\gamma \E_{x'}[J_k^{\pi}(x')|x,\abold]],\nonumber\\
        &=\E_{\abold}\left[\sum_{n\in\N}\E_{w_n}[r_n(x_n,\alpha_n|w_n)]+\gamma \E_{x'}[\Tb^{\pi}J_{k-1}(x'|W=\mathbf{1}) + (1-\beta)\Tb^{\pi}J_{k-1}(x'|W=\mathbf{0})|x,\abold]\right],\nonumber\\
        \begin{split}
            &= \E_{\abold}[\beta r(x,\abold,|W = \mathbf{1}) + (1-\beta) r(x,\abold,|W = \mathbf{0})\\
            &\quad+ \gamma \E_{x'}[\beta \E_{\abold'}[r(x',\abold'|W = \mathbf{1}) + \gamma\E_{x''}[J_{k-1}^{\pi}(x''|W = \mathbf{1})|x',\abold')]] \\
            &\quad +(1-\beta) \E_{\abold'}[r(x',\abold'|W = \mathbf{0}) + \gamma\E_{x''}[J_{k-1}^{\pi}(x''|W = \mathbf{0})|x',\abold')]] |x,\abold]],\nonumber
        \end{split}\\
        \begin{split}
            &=\E_{\abold}[\beta r(x,\abold,|W = \mathbf{1})+\gamma \E_{x'}[\beta \E_{\abold'}[r(x',\abold'|W = \mathbf{1}) + \gamma\E_{x''}[J_{k-1}^{\pi}(x''|W = \mathbf{1})|x',\abold')]]|x,\abold]]\\
            &\quad +\E_{\abold}[(1-\beta) r(x,\abold,|W = \mathbf{0})+\gamma \E_{x'}[(1-\beta) \E_{\abold'}[r(x',\abold'|W = \mathbf{0}) + \gamma\E_{x''}[J_{k-1}^{\pi}(x''|W = \mathbf{0})|x',\abold')]]|x,\abold],\nonumber
        \end{split}\\
        \begin{split}
            &= \beta \E_{\abold}[r(x,\abold,|W = \mathbf{1}) + \gamma \E_{x'}[J_{k}^{\pi}(x'|W = \mathbf{1})|x,\abold]]\\
            &\quad+ (1-\beta) \E_{\abold}[r(x,\abold,|W = \mathbf{0}) + \gamma \E_{x'}[J_{k}^{\pi}(x'|W = \mathbf{0})|x,\abold]],\nonumber
        \end{split}\\
        &=\beta \Tb^{\pi}J_{k}(x|W=\mathbf{1}) + (1-\beta)\Tb^{\pi}J_{k}(x|W=\mathbf{0})\label{induction k}
    \end{align}
    \else
    \begin{align}
        &J_{k+1}^{\pi}(x)\nonumber \\
        &= \E_{\abold}[r^R(x,\abold)+\gamma \E_{x'}[J_k^{\pi}(x')|x,\abold]],\nonumber\\
        &=\E_{\abold}[\E_{W}[r(x,\abold|W)]+\gamma \E_{x'}[J_k^{\pi}(x')|x,\abold]],\nonumber\\
        \begin{split}
             &=\E_{\abold}\Bigg[\sum_{n\in\N}\E_{w_n}[r_n(x_n,\alpha_n|w_n)]\\
             &\quad +\gamma \E_{x'}[\Tb^{\pi}J_{k-1}(x'|W=\mathbf{1})\\
             &\quad + (1-\beta)\Tb^{\pi}J_{k-1}(x'|W=\mathbf{0})|x,\abold]\Bigg],\nonumber
        \end{split}\\
        \begin{split}
            &= \E_{\abold}[\beta r(x,\abold,|W = \mathbf{1}) + (1-\beta) r(x,\abold,|W = \mathbf{0})\\
            &\quad+ \gamma \E_{x'}[\beta \E_{\abold'}[r(x',\abold'|W = \mathbf{1}) \\
            &\quad \quad+ \gamma\E_{x''}[J_{k-1}^{\pi}(x''|W = \mathbf{1})|x',\abold')]] \\
            &\quad +(1-\beta) \E_{\abold'}[r(x',\abold'|W = \mathbf{0}) \\
            &\quad \quad+ \gamma\E_{x''}[J_{k-1}^{\pi}(x''|W = \mathbf{0})|x',\abold')]] |x,\abold]],\nonumber
        \end{split}\\
        \begin{split}
            &=\E_{\abold}[\beta r(x,\abold,|W = \mathbf{1})\\
            &\quad +\gamma \E_{x'}[\beta \E_{\abold'}[r(x',\abold'|W = \mathbf{1}) \\
            &\quad \quad + \gamma\E_{x''}[J_{k-1}^{\pi}(x''|W = \mathbf{1})|x',\abold')]]|x,\abold]]\\
            &\quad +\E_{\abold}[(1-\beta) r(x,\abold,|W = \mathbf{0})\\
            &\quad \quad +\gamma \E_{x'}[(1-\beta) \E_{\abold'}[r(x',\abold'|W = \mathbf{0})\\
            &\quad \quad + \gamma\E_{x''}[J_{k-1}^{\pi}(x''|W = \mathbf{0})|x',\abold')]]|x,\abold],\nonumber
        \end{split}\\
        \begin{split}
            &= \beta \E_{\abold}[r(x,\abold,|W = \mathbf{1}) + \gamma \E_{x'}[J_{k}^{\pi}(x'|W = \mathbf{1})|x,\abold]]\\
            &\quad+ (1-\beta) \E_{\abold}[r(x,\abold,|W = \mathbf{0})\\
            &\quad \quad + \gamma \E_{x'}[J_{k}^{\pi}(x'|W = \mathbf{0})|x,\abold]],\nonumber
        \end{split}\\
        &=\beta \Tb^{\pi}J_{k}(x|W=\mathbf{1}) + (1-\beta)\Tb^{\pi}J_{k}(x|W=\mathbf{0})\label{induction k}
    \end{align}
    \fi

    Taking $k\to\infty$:
    \begin{align*}
         &J^{\pi}(x) \\
         &= \lim_{k\to\infty}J^{\pi}_k(x)\\
         &=\left[\beta \lim_{k\to\infty}J^{\pi}_k(x|W = \mathbf{1})+(1-\beta)\lim_{k\to\infty}J^{\pi}_k(x|W = \mathbf{0})\right]\\
         &= \beta J^{\pi}(x|W = \mathbf{1})+(1-\beta)J^{\pi}(x|W = \mathbf{0})
    \end{align*}

    By Theorem \ref{dropout value function theorem},
    \begin{align}
        V_R^{\pi}(x) &= \E_W\E_{x_{-\N_W}}J^{\pi}(x)\nonumber\\
        &= \E_W\E_{x_{-\N_W}}[\beta J^{\pi}(x|W = \mathbf{1}) + (1-\beta)J^{\pi}(x|W = \mathbf{0})]\nonumber\\
        &= \E_W\E_{x_{-\N_W}}\left[\beta J^{\pi}(x|W = \mathbf{1}) + \frac{1-\beta}{1-\gamma}\rbar\right]\label{geo series}\\
        &= \beta \E_W V^{\pi}(x|W = \mathbf{1})+ \frac{1-\beta}{1-\gamma}\rbar\nonumber
    \end{align}
    
    As $r(x,\abold|W = \mathbf{0})= 0)$, the term $(1-\beta)\Tb^{\pi}V_0(x|W=\mathbf{0})$ in \eqref{induction 0} becomes $0$ and the term $(1-\beta)\Tb^{\pi}V_k(x|W=\mathbf{0})$ in \eqref{induction k} becomes  $0$. Equation \ref{geo series} follows.
    \fi
\end{proof}

The next lemma establishes the optimality gap produced by controlling the pre-dropout system with the optimal robust policy. This is the maximum loss in error accrued by controlling the system with the robust policy if dropout never occurred. 

\begin{lemma}\label{opt gap}
Controlling the pre-dropout MDP with $\pi_R^*$ yields,
\ifarxiv
\begin{equation}
    V^*(x|W = \mathbf{1}) - V^{\pi_R^*}(x) \leq (1-\beta^N)[V^*(x|W=\mathbf{1})- V^{\pi_R^*}(x|W = \mathbf{1})].
\end{equation}
\fi
\end{lemma}

\begin{proof}  Let $\pi^\dagger_W(\abold|x) = \pi_R^*(\abar|\xbar)\prod_{n\notin\N_W}\pi_n^*(\abold_n|x)$. Decompose the error as:
    \ifarxiv
    \begin{align*}
        &V^*(x) - V^{\pi_R^*}(x)\\
        &= V^*(x|W = \mathbf{1}) - V^{\pi_R^*}(x|W = \mathbf{1}),\\
        &=V^*(x|W = \mathbf{1}) -\E_{W'} V^{\pi^\dagger_{W'}}(x|W = \mathbf{1}),\\
        &=\E_{W'}[V^*(x|W = \mathbf{1}) - V^{\pi^\dagger_{W'}}(x|W = \mathbf{1})],\\
        &=\E_{W'}\left[\sum_{n\in \N_{W'}}V^*(x_n|W = \mathbf{1}) + \sum_{n\notin\N_{W'}}V^*(x_n|W = \mathbf{1}) - \sum_{n\in\N_{W'}}V^{\pi^\dagger_{W'}}(x_n|W = \mathbf{1}) - \sum_{n\notin\N_{W'}}V^{\pi^\dagger_{W'}}(x_n|W = \mathbf{1})\right].
    \end{align*}
    \fi
    Using separable properties of factored MDPs, terms may be separated and canceled.
    The difference is:
    \ifarxiv
    \begin{align*}
        &\E_{W'}\left[\sum_{n\in \N_{W'}}V^*(x_n|W = \mathbf{1})+ \sum_{n\notin\N_{W'}}V^*(x_n|W = \mathbf{1}) - \sum_{n\in\N_{W'}}V^{\pi^\dagger_{W'}}(x_n|W = \mathbf{1})- \sum_{n\notin\N_{W'}}V^{\pi^\dagger_{W'}}(x_n|W = \mathbf{1})\right]\\
        &=\E_{W'}\left[\sum_{n\in\N_{W'}}V^*(x_n|W = \mathbf{1}) - \sum_{n\in\N_{W'}}V^{\pi^\dagger_{W'}}(x_n|W = \mathbf{1})\right]\\
        &\leq\E_{W'}\left[\sum_{n\in\N_{W'}}V^*(x_n|W = \mathbf{1}) - \sum_{n\in\N_{W'}}V^{\pi_R^*}(x_n|W = \mathbf{1})\right]
    \end{align*}
    \fi
    as $\pi_{R,n}^*$ is out-performed by $\pi^*$ for $W=\mathbf{1}$. 
    \vspace{1mm}
     The value difference \bluetoo{expression is equal to}:
     \ifarxiv
    \begin{align*}
        &\E_{W'}\sum_{n\in\N_{W'}}\left[V^*(x_n|W = \mathbf{1}) - V^{\pi_R^*}(x_n|W = \mathbf{1})\right]\\
        &=\sum_{W'}P(W')\sum_{n\in\N_{W'}}\left[V^*(x_n|W = \mathbf{1}) - V^{\pi_R^*}(x_n|W = \mathbf{1})\right]\\
        &=\sum_{k = 1}^NP(|W' = 1| = k)\sum_{\{W'|\ |W' = 1| =k\}}\sum_{n\in\N_{W'}}\left[V^*(x_n|W = \mathbf{1}) - V^{\pi_R^*}(x_n|W = \mathbf{1})\right]\\
        &=\sum_{k = 1}^N\beta^{k}(1-\beta)^{N-k}{N\choose k}\frac{k}{N}\left[V^*(x|W = \mathbf{1}) - V^{\pi_R^*}(x|W = \mathbf{1})\right]\\
        &=\sum_{k = 1}^N\beta^{k}(1-\beta)^{N-k}{N-1\choose k-1}\left[V^*(x|W = \mathbf{1}) - V^{\pi_R^*}(x|W = \mathbf{1})\right]\\
        &=\sum_{k=0}^{N-1}{N-1\choose k}\beta^{k}(1-\beta)^{N-k}[V^*(x|W=\mathbf{1})- V^{\pi_R^*}(x|W = \mathbf{1})]\\
        &=[V^*(x|W=\mathbf{1})- V^{\pi_R^*}(x|W = \mathbf{1})](1-\beta^N)
    \end{align*}
    \fi
    \vspace{-3mm}\end{proof}

\vspace{-3mm}\subsection{Suboptimality of Pre-Dropout Policy}\label{sec:subopt pre}
Given the structural relationship between the pre- and post-dropout MDPs, a natural question is if this relationship extends to value optimality (without defining the robust MDP). To demonstrate this cannot hold in general, we present a counter example. Consider a system with optimal value $V^*(x|W = \mathbf{1})$ where the post-dropout rewards are defined as the marginalization $r(x,\abold|W) = \E_{\xbar,\abar}r(x,\abold|W = \mathbf{1})$.  Then, $V(x|W) = \E_{x_{-\N_W}}V(x|W = \mathbf{1})$.

Unfortunately, the value $V^*(x|W\neq \mathbf{1})$ cannot be calculated as $\E_{x_{-\N_W}}[V^*(x|W = \mathbf{1})]$. To verify this, define $b_n \equiv w_n$ to force the desired realization. Then,
\begin{equation}
    V_R(x) = V(x|W) = \E_{-\N_W} V(x|W = \mathbf{1})
\end{equation}

To see this, note that the Bellman optimality criterion necessitates that $V^*(s) = \mathbf{T}V^*(s) = \max_{\pi} \mathbb{E}_{\abold\sim \pi}[r(s,\abold)+\gamma V^*(s')]$. Therefore, if $\tilde{V}$ are optimal values, then,
\ifarxiv
\begin{align}
    &\tilde{V}(x|W\neq\mathbf{1})\\
    &=\mathbf{T}\tilde{V}(x|W\neq\mathbf{1}),\nonumber\\
    &=\mathbf{T}\E_{-\N_W}  V^*(x|W = \mathbf{1}),\nonumber\\
    &=\max_{\pi}\E_{\abold\sim\pi}\E_{-\N_W}[r(x,\abold)+\gamma \E_{x'} [V^*(x'|W = \mathbf{1})|x,\abold]]. \label{bellman}
\end{align}
\fi

However, we find that,
\ifarxiv
\begin{align}
    &\tilde{V}(x|W\neq\mathbf{1})\\
    &=\E_{-\N_W} V^*(x|W = \mathbf{1}),\nonumber\\
    &=\E_{-\N_W}\E_{\abold\sim\pi^*}[r(x,\abold)+\gamma \E_{x'}[V^*(x|W = \mathbf{1})|x,\abold]],\nonumber\\
    &=\E_{-\N_W}\max_{\pi}\E_{\abold\sim\pi}[r(x,\abold)+\gamma \E_{x'}[V^*(x|W = \mathbf{1})|x,\abold]].\label{margeq}
\end{align}
\fi

\normalsize
Clearly, \eqref{bellman} and \eqref{margeq} are not guaranteed to coincide as the maximization and marginalization operations do not commute. This calculation has merely evaluated the value of the \emph{new} system under the policy developed for the \emph{old} system. Therefore, one possibility is to evaluate $V^{\pi}(s)$ for other policies $\pi$ and find one such that the principle of optimality holds. In practice, however, it is more reasonable to use the robust MDP formulation to automatically relate all possible realizations of the system.

\section{Model-Free Policy Evaluation}\label{sec: policy evaluation}

\subsection{Method} \label{method}
Given Theorem \ref{marginalization theorem}, policy IS can now be adapted for the objective of estimating values of policies for the robust MDP and for specific post-dropout realizations based on trajectories generated by the pre-dropout system. 

The first step is to represent the desired post-dropout policy $\phi$ in a usable format that satisfies \eqref{agument}. For the robust MDP, the desired policy $\phi$ can be used as-is because the state and action spaces between the compared systems are identical. To perform policy evaluation for a realization $(\mm|W)$, however, the post-dropout policy $\phi'(\abar|\xbar)$ must be augmented to be  a function of $\abold$ and $x$. As $\phi'$ is independent of $-\N_W$, the policies for the removed agents can be augmented with pre-dropout distributions as in \eqref{agument}.

Next, trajectories are generated on the pre-dropout MDP using $\pi$, and the policy IS estimate can be formed. By performing both the sampling and the policy IS routines on the pre-dropout system, the issue of non-cancellation of the transition probabilities in \eqref{ratiobad} is resolved. The estimate can then be transformed to the post-dropout value via marginalization according to Theorem \ref{dropout value function theorem}:

\ifarxiv
\begin{align}
    V_H^{\phi}(x|W) &= \E_{x_{-\N_W}}\E_{\tau\sim q}\left[\frac{p(\tau)}{q(\tau)}\sum_{t=0}^{H-1} \gamma^{t} r_t^R(x_t,\abold_t) \right].\label{policy is vw}
\end{align}
\fi

This reweighted sample return may be constructed via any standard policy IS technique, such as the step-wise estimator, weighted estimator, or doubly robust estimator \cite{jiang2016doubly}:

\ifarxiv
\begin{align}
    \Jhat(x) &= \frac{1}{|D|}\sum_{i=1}^{|D|}\left[\frac{p(\tau_i)}{q(\tau_i)}\sum_{t=0}^{H-1} \gamma^{t} r_t^R(x_t,\abold_t) \right],
\end{align}
\fi
where $p(\tau_i)/q(\tau_i) = \prod_{t=0}^{H-1}\phi(\abold_t^{(i)}|x_t^{(i)})/\pi(\abold_t^{(i)}|x_t^{(i)})$, a superscript $(i)$ and subscript $t$ means the state or action taken at time $t$ in trajectory $i$, and $|D|$ is the size of the used dataset. Under the assumption of bounded rewards in Assumption \ref{as:reward} and full support in Assumption \ref{support}, the IS estimator will be bounded by some maximum value $\Jmaxt$.

The marginalization step may not be computed with respect to the true stationary distribution as it assumed that the true transition matrix is unknown. However, the empirical stationary distribution $\muhat(x_{-\N_W})$ may be estimated and used in place of the true distribution. To estimate a stationary distribution of a Markov chain empirically, it is beneficial to use one trajectory with a long horizon to achieve the Markov chain's mixing time \blue{\cite{wolfer2020mixing}}; this is the basis of Markov chain Monte Carlo techniques. To this end, we suggest generating a trajectory $x_1,\dots,x_{H_{\mu}}$ of length $H_{\mu}$ separate from the trajectories used for the IS routine. Define the shorthand notation $\xnot = \{x'|x'_{-\N_W} = x_{-\N_W}\}$ to be the set of states whose substates for agents $-\N_W$ match those in the given state $x$. The \emph{empirical stationary distribution} is,
\ifarxiv
\begin{align}
    \muhat(x) &= H_{\mu}^{-1}\sum_{i=0}^{H_{\mu}-1}\mathbbm{1}[x_i = x].\label{muhat}
\end{align}
Similarly to \eqref{mu big}, the empirical distribution  $\muhat(x_{-\N_W})$ is the summation,
\begin{align}
    \muhat(x_{-\N_W}) &= \sum_{\xnot}\muhat(x). \label{muhat big}
\end{align}
\fi

To finish the value estimate of the robust MDP, a final marginalization step over $W$ may be completed under the assumption that the dropout probabilities are known \emph{a priori}.
\begin{align}
    V^{\phi}_{RH}(x) &= \E_W[V^{\phi}_H(x|W)]. \label{policy is vr}
\end{align}
The estimators can similarly be constructed as,
\ifarxiv
\begin{align}
    \Vhat_H^{\phi}(x|W) &= \sum_{x_{-\N_W}}\muhat(x_{-\N_W}) \Jhat(x),\nonumber\\
    &=\sum_{x_{-\N_W}}\muhat(x_{-\N_W}) \Jhat(x_{\N_W}, x_{-\N_W}),\label{vhat split jhat}\\
    \Vhat^{\phi}_{RH}(x) &= \sum_W p(W)\Vhat_H^{\phi}(x|W),\nonumber\\
    &=\sum_W \prod_{n=1}^N \beta_n^{w_n}(1-\beta_n)^{1-w_n}\Vhat_H^{\phi}(x|W).\label{vhatrh}
\end{align}
\fi

To compute $\Vhat_H^{\phi}(x|W)$, note that $\Jhat_W(x)$ will need to be known for all $\{x'|x'_{\N_W} = x_{\N_W}\}$ as evident in \eqref{vhat split jhat}. Computing $\Vhat^{\phi}_{RH}(x)$ will thus require calculation of  $\Jhat_W(x)$ for all $x$. If $|D|$ trajectories are used to compute each estimate, a total of $|\X||D|$ total trajectories will be needed for policy IS.

If $\beta_n\equiv \beta$, the following simplification can be made to \eqref{vhatrh}:
\begin{align}
    \Vhat^{\phi}_{RH}(x) &=\sum_W \beta^{|W=1|}(1-\beta)^{|W=0|}\Vhat_H^{\phi}(x|W).
\end{align}

 Given the described policy evaluation technique, policy search can be implemented according to the parameters of the application. A key benefit of the policy IS technique is that it resolves the conflicting objectives of controlling the existing system for good value while evaluating post-dropout policies. As discussed in Section \ref{sec:subopt pre}, optimality of the pre- and post-dropout systems may be unrelated, so the pre-dropout system should not necessarily be controlled with the optimal post-dropout policy. By selecting behavioral policies that produce good pre-dropout value, while evaluating target policies for the robust or post-dropout models, good policy evaluation and system execution can be completed. 

\subsection{Performance}
In this section the performance of the estimator $\Vhat^{\phi}_{RH}$ will be analyzed.
\begin{lemma}\label{vmax lemma}
    Define,
    \ifarxiv
    \begin{align}
        V^{\max}_H = \frac{1-\gamma^H}{1-\gamma}\max_{\abold, x} r^R(x,\abold).
    \end{align}
    \fi
    Then $V^{\max}_H\geq J_H(x)$.
\end{lemma}

\ifarxiv
The following result (Prop. 2.19 from \cite{paulin2015concentration}), provides a concentration bound on the convergence of the empirical stationary distribution. This bound depends on the \emph{mixing time} of the MDP, which is a structural property defined as,
\begin{align}
    &d(t)\triangleq \sup_{x\in\X}d_{TV}(P^t(x,\cdot),\mu),\\
    &t_{mix}\triangleq t_{mix}(1/4).
\end{align}
\fi
This property is a measure of the time required of a Markov chain for the distance to stationarity to be small. 

\vspace{1mm}
\begin{lemma} \label{paulin lemma}
    Consider a uniformly ergodic Markov chain with a countable state space, unique stationary distribution, and mixing time $t_{mix}$. For any $\delta\geq 0$,
    \vspace{1mm}
    \small\begin{equation}
        P(|d_{TV}(\muhat, \mu) - \E_{\mu} [d_{TV}(\muhat, \mu)]|\geq \delta)\leq 2\exp(-\delta^2H_{\mu}/(4.5 t_{mix})).
    \end{equation}
\end{lemma}
\normalsize

Under the assumptions of a finite state space in the problem formulation and ergodicity in Assumption \ref{ergodic ass}, we will satisfy the necessary conditions to apply Lemma \ref{paulin lemma}. 
\begin{lemma} \label{bounded exp}
    For stationary chains, the expected total variational distance between the empirical distributional and the stationary distribution is bounded and inversely proportional to $H_{\mu}$:  $\E_{\mu} [d_{TV}(\muhat, \mu)]\leq C(H_{\mu})$.
\end{lemma}
\vspace{1mm}
\begin{proof}
    See Proposition 3.16 in \cite{paulin2015concentration}.
\end{proof}
Note that $\E_{\mu} [d_{TV}(\muhat, \mu)]\to 0$ as $H_{\mu}\to\infty$.
\begin{remark}
    If an unbiased policy IS technique is used, then $\Vhat^{\phi}_{RH}(x)$ is an asymptotically unbiased estimator for $V^{\phi}_{RH}(x)$ as $H_{\mu}\to\infty$.
\end{remark}

\vspace{1mm}
\begin{lemma} \label{marg error}
    \textbf{(Error Produced by Empirical Marginalization)} Let $t_{mix}$ be the mixing time of $\mm$. Let $\muhat$ be the empirical stationary distribution computed from a trajectory of length $H_{\mu}$. The error produced by empirical marginalization can be bounded by,
    
    \ifarxiv
    \begin{align}
        P\left(\left|\E_{x_{-\N_W}}J_H(x)-\sum_{x_{-\N_W}}\muhat(x_{-\N_W})J_H(x)\right|>\epsilon\right)\leq 2\exp\left(-\left(\frac{\epsilon}{|\X_n|^{|W=1|}\sum_{x_{-\N_W}}J_H(x)}-C(H_{\mu})\right)^2\frac{H_{\mu}}{4.5t_{mix}}\right).
    \end{align}
    \fi
\end{lemma}
\begin{proof}
    \ifarxiv
    \begin{align}
        &P\left(\left|\E_{x_{-\N_W}}J_H(x)-\sum_{x_{-\N_W}}\muhat(x_{-\N_W})J_H(x)\right|>\epsilon\right)\nonumber\\
        &=P\left(\left|\sum_{x_{-\N_W}}\mu(x_{-\N_W})J_H(x)-\sum_{x_{-\N_W}}\muhat(x_{-\N_W})J_H(x)\right|>\epsilon
        \right),\nonumber\\
        &\leq P\left(\sum_{x_{-\N_W}}J_H(x)\left|\mu(x_{-\N_W})-\muhat(x_{-\N_W})\right|>\epsilon
        \right),\nonumber\\
        &= P\left(\sum_{x_{-\N_W}}J_H(x)\left|\sum_{\xnot}\mu(x)-\sum_{\xnot}\muhat(x)\right|>\epsilon
        \right),\nonumber\\
        &\leq P\left(\sum_{x_{-\N_W}}J_H(x)\sum_{\xnot}\left|\mu(x)-\muhat(x)\right|>\epsilon
        \right),\nonumber\\
        &\leq P\left(\sum_{x_{-\N_W}}J_H(x)|\X_n|^{|W=1|}d_{TV}(\mu,\muhat)>\epsilon
        \right),\nonumber\\
        &= P\left(d_{TV}(\mu,\muhat)>\frac{\epsilon}{|\X_n|^{|W=1|}\sum_{x_{-\N_W}}J_H(x)}
        \right),\nonumber\\
        &= P\left(d_{TV}(\mu,\muhat)-\E_{\mu} d_{TV}(\mu,\muhat)>\frac{\epsilon}{|\X_n|^{|W=1|}\sum_{x_{-\N_W}}J_H(x)}-\E_{\mu} d_{TV}(\mu,\muhat)
        \right),\nonumber\\
        &\leq P\left(d_{TV}(\mu,\muhat)-\E_{\mu} d_{TV}(\mu,\muhat)>\frac{\epsilon}{|\X_n|^{|W=1|}\sum_{x_{-\N_W}}J_H(x)}-C(H_{\mu})
        \right),\nonumber\\
        &\leq P\left(\left|d_{TV}(\mu,\muhat)-\E_{\mu} d_{TV}(\mu,\muhat)\right|>\frac{\epsilon}{|\X_n|^{|W=1|}\sum_{x_{-\N_W}}J_H(x)}-C(H_{\mu})
        \right)\label{last marg}
    \end{align}
    \fi

    Applying Lemma \ref{paulin lemma} to \eqref{last marg}, the final result can be obtained. 
    \end{proof}

\begin{theorem}\label{theorem realization}
\textbf{(Performance of Policy IS Estimate of Realized System)} Let $\Vhat_{RH}^{\pi}(x|W)$ be the estimated value of $(\mm|W)$ formed by estimating $\Jhat_H$ with policy IS and then marginalizing with respect to $\muhat(x_{-\N_W})$. Let $V_R^{\pi}(x|W)$ be the corresponding true value. Let the selected IS estimator have bounded bias $|J_H(x)-\E J_H(x)|\leq B_{IS}(H)$ and use a dataset of $|D|$ i.i.d. trajectories each of length $H$ for each $x$. Let the empirical stationary distribution $\hat\mu(x_{-\N_W})$ be formed from an additional trajectory of length $H_{\mu}$. Let $\rmax\triangleq\max_{x,\abold} r^R(x,\abold)$, and let $\epsilon'= \frac{\gamma^H}{1-\gamma}\rmax +B_{IS}(H)$. Then for $\delta\geq0$,

\ifarxiv
\begin{align}
&P(|V_R^{\pi}(x|W) - \Vhat_{RH}^{\pi}(x|W)|\geq \delta+\epsilon')\nonumber\\
&\leq 2\left(\exp\left(-\left(\frac{\delta}{2|\X_n|^{|W=1|}\sum_{x_{-\N_W}}J_H(x)}-C(H_{\mu})\right)^2\frac{H_{\mu}}{4.5t_{mix}}\right) + \exp\left(\frac{-|D|\delta^2}{4\Jmaxt^2} \right)\right).\label{bound1}
\end{align}
\fi
\end{theorem}

\begin{proof}
The error can be decomposed as,
    \ifarxiv
    \begin{align*}
        &V_R^{\pi}(x|W) - \Vhat_{RH}(x|W)\\
        &=\E_{x_{-\N_W}}J(x)-\sum_{x_{-\N_W}}\muhat(x_{-\N_W})\Jhat_H(x),\\
        &=\E_{x_{-\N_W}}\left[J(x)-J_H(x)+J_H(x)\right]+\sum_{x_{-\N_W}}\muhat(x_{-\N_W})\left(-J_H(x)+J_H(x)-\E J_H(x) + J_H(x) - \Jhat_H(x)\right),\\
        \begin{split}
            &=\E_{x_{-\N_W}}[J(x)-J_H(x)]+\left[\E_{x_{-\N_W}}J_H(x)-\sum_{x_{-\N_W}}\muhat(x_{-\N_W})J_H(x)\right]\\
            &\quad+\sum_{x_{-\N_W}}\muhat(x_{-\N_W})(J_H(x)-\E J_H(x)) + \sum_{x_{-\N_W}}\muhat(x_{-\N_W})(\E J_H(x) - \Jhat_H(x)),
        \end{split}\\
        &\leq\Delta_H + \Delta_\mu + \Delta_B + \Delta_{IS},
    \end{align*}
    \fi

\normalsize
\noindent where $\Delta_H$ is the error of a trajectory of length $H$ to estimate the infinite horizon reward, $\Delta_{\mu}$ is the error due to empirical marginalization, $B_{IS}$ is the upper bound of the (possible) bias of the IS estimator, and $\Delta_{\text{IS}}$ is the error of the selected IS estimator. For any value function it is known that $|J(x)-J_H(x)|\leq \rmax \gamma^H/(1-\gamma) = \tilde{r}$ so $\Delta_H$ can be deterministically bounded. The stochastisity in the error therefore comes from the importance sampling step. With substitution, the triangle inequality, and noting that $\tilde{r}\geq 0$ and $B_{IS}\geq 0$,
\ifarxiv
\begin{align*}
    &P(|\Delta_{\text{H}}+\Delta_{\mu}+B_{IS}+ \Delta_{\text{IS}}|\geq \epsilon)\\
    &\leq P(|\Delta_{\mu}+\Delta_{\text{IS}}| \geq \epsilon-\tilde{r}-B_{IS}).
\end{align*}
\fi

Define $\epsilon  = \delta +\tilde{r} + B_{IS}(H)$. Then,
\ifarxiv
\begin{align*}
    &=P(|\Delta_{\mu}+\Delta_{\text{IS}}| \geq \epsilon-\tilde{r}-B_{IS}(H))\\
    &=P(|\Delta_{\mu}+\Delta_{\text{IS}}| \geq \delta)\\
    &\leq P\left(|\Delta_{\mu}|\geq\frac{1}{2}\delta\right) + P\left(|\Delta_{IS}|\geq\frac{1}{2}\delta\right).
\end{align*}
\fi

Next, bounding $\Delta_{\mu}$ can be accomplished via Lemma \ref{marg error}. 
\ifarxiv
\begin{align*}
    &P(|\Delta_{\mu}|\geq \epsilon_0)\\
    &=P\left(\left|\E_{x_{-\N_W}}J_H(x)-\sum_{x_{-\N_W}}\muhat(x_{-\N_W})J_H(x)\right|>\frac{1}{2}\delta\right),\\
    &\leq 2\exp\left(-\left(\frac{\delta}{2|\X_n|^{|W=1|}\sum_{x_{-\N_W}}J_H(x)}-C(H_{\mu})\right)^2\frac{H_{\mu}}{4.5t_{mix}}\right).
\end{align*}
\fi
Since the rewards are bounded and given Assumption \ref{support}, the estimates of $J$ produced by policy IS are bounded by $\Jmaxt$.
With i.i.d. samples of bounded random variables, a Hoeffding confidence bound may be used:
    \ifarxiv
    \begin{align*}
        &P(|\Delta_{IS}|>\epsilon_0)\\
        &=P\left(\sum_{x_{-\N_W}}\muhat(x_{-\N_W})|\E J_H(x) - \Jhat_H(x)|>\epsilon_0\right),\\
        &\leq P\left(\max_{x}|\E J_H(x) - \Jhat_H(x)|>\epsilon_0\right).
    \end{align*}
    \fi
    By the Hoeffding bound, every $P\left(|\E J_H(x) - \Jhat_H(x)|>\epsilon_0\right)\leq 2\exp\left(\frac{-|D|\epsilon_0^2}{\Jmaxt^2} \right)$. As this bounds holds for every $x$, it will hold for the max:
    \begin{align*}
        P\left(\max_{x}|\E J_H(x) - \Jhat_H(x)|>\epsilon_0\right)&\leq 2\exp\left(\frac{-|D|\epsilon_0^2}{\Jmaxt^2} \right).
    \end{align*}
    \ifarxiv
    Therefore,
    \begin{align*}
        P\left(|\Delta_{IS}|>\frac{1}{2}\delta\right)&\leq 2\exp\left(\frac{-|D|\delta^2}{4\Jmaxt^2} \right).
    \end{align*}
    \fi
\end{proof}

\bluethree{Theorem \ref{theorem realization} addresses Problem 1 by showing that the values of post-dropout policies for specific realizations may be effectively estimated. The estimator is constructed from finite trajectories subject to IS and empirical marginalization, and can be described by an exponential confidence interval.} The first term arises due to marginalization and the second due to the IS estimator. The IS error will go to zero as $|D|\to\infty$, and the marginalization error will go to zero as $H_{\mu}\to\infty$. Note that $H \ll H_{\mu}$ is beneficial as the variance of standard IS estimators increases rapidly with $H$, but $H_{\mu}$ must be large for the empirical distribution to reflect the true stationary distribution. 

This bound justifies the marginalized IS estimator as a valid technique as $\Vhat_R(x|W)$ is exponentially concentrated around $V_R(x|W)$. It is possible to obtain more sophisticated bounds for the IS estimator (see \cite{thomas2015high}), but these bounds are often pessimistic for practical applications. A discussion of normal approximations for IS estimator bounds can be found in \cite{thomas2016data}.

\begin{theorem} \label{theorem exp}
    \textbf{(Performance of Policy IS Estimate of Robust System)} Let $\Vhat_{RH}^{\pi_R}(x)$ be the estimated value of the robust MDP for policy $\pi_R$, and let $V_R^{\pi_R}(x)$ be the corresponding true value. Let the selected IS estimator have bounded bias $|J_H(x)-\E J_H(x)|\leq B_{IS}(H)$ and use a dataset of $|D|$ i.i.d. trajectories each of length $H$ for each $x$. Let the empirical stationary distribution $\muhat$ be formed from an additional trajectory of length $H_{\mu}$. Let $\rmax\triangleq\max_{x,\abold} r^R(x,\abold)$, and let $\epsilon'= \frac{\gamma^H}{1-\gamma}\rmax +B_{IS}(H)$. Then for $\delta\geq0$,
    
    \ifarxiv
    \begin{align}
        &P(|V_R^{\pi_R}(x) - \Vhat_{RH}^{\pi_R}(x)|>\delta+\epsilon')\nonumber\\
        &\leq 2\Bigg(\exp\left(-\left(\frac{\delta}{2|\X|V^{\max}_H}-C(H_{\mu})\right)^2\frac{H_{\mu}}{4.5t_{mix}}\right)+\exp\left(\frac{-|D|\delta^2}{4\Jmaxt^2} \right)\Bigg).
    \end{align}
    \fi
\end{theorem}
\begin{proof}
    The steps of this proof follow similarly to those of the proof of Theorem \ref{theorem realization} except now with the outer expectation over $W$. As it is assumed the distribution for $W$ is known, no additional error should be incurred from that step.
    \ifarxiv
    \begin{align*}
        &|V_R^{\pi_R}(x) - \Vhat_{RH}^{\pi_R}(x)|\\
        &=\E_W|V_R^{\pi_W}(x|W) - \Vhat_{RH}^{\pi_W}(x|W)|,\\
        &\leq\E_W|\Delta_H^W + \Delta_\mu^W + \Delta_B^W + \Delta_{IS}|.
    \end{align*}
    \fi

    Note that $\tilde{r}$ and $B_{IS}(H)$ are both deterministic upper bounds over all $W$. Then,
    \ifarxiv
    \begin{align*}
        &P(\E_W|\Delta_{\text{H}}^W+\Delta_{\mu}^W+\Delta_B^W+ \Delta_{\text{IS}}^W|\geq \epsilon)\\
        &\leq P(\E_W|\tilde{r}| + \E_W|B_{IS}|+ \E_W|\Delta_{\mu}^W+\Delta_{\text{IS}}^W|\geq\epsilon),\\
        &=P(\E_W|\Delta_{\mu}^W+\Delta_{\text{IS}}^W| \geq \delta),\\
        &\leq P\left(\max_W|\Delta_{\mu}^W|\geq \frac{1}{2}\delta\right)+P\left(\max_W|\Delta_{IS}^W|\geq \frac{1}{2}\delta\right).
    \end{align*}
    \fi
    The established Hoeffding bound holds for all $\Delta_{IS}^W$: $P\left(|\Delta_{IS}^W|>\frac{1}{2}\delta\right)\leq 2\exp\left(\frac{-|D|\delta^2}{4\Jmaxt^2} \right)$. Therefore, 
    \begin{align*}
        P\left(\max_W|\Delta_{IS}^W|>\frac{1}{2}\delta\right)&\leq 2\exp\left(\frac{-|D|\delta^2}{4\Jmaxt^2} \right).
    \end{align*}
    To bound the error from marginalization, first note that $\sum_{x_{-\N_W}}J_H(x)\leq |\X_n|^{|W=0|}V^{\max}_H$ by Lemma \ref{vmax lemma}. Then,
    \ifarxiv
    \begin{align}
        P\left(\left|\Delta_{\mu}^W\right|>\frac{1}{2}\delta\right)&\leq 2\exp\left(-\left(\frac{\delta}{2|\X_n|^{|W=1|}\sum_{x_{-\N_W}}J_H(x)}-C(H_{\mu})\right)^2\frac{H_{\mu}}{4.5t_{mix}}\right),\\
        &\leq 2\exp\left(-\left(\frac{\delta}{2|\X|V^{\max}_H}-C(H_{\mu})\right)^2\frac{H_{\mu}}{4.5t_{mix}}\right).
    \end{align}
    \fi
    As the bound holds for all $|\Delta_{\mu}^W|$, it will hold for the maximum.
    \ifarxiv
    \begin{align*}
         P\left(\max_W|\Delta_{IS}^W|\geq \frac{1}{2}\delta\right) &\leq 2 \exp\left(-\left(\frac{\delta}{2|\X|V^{\max}_H}-C(H_{\mu})\right)^2\frac{H_{\mu}}{4.5t_{mix}}\right).
    \end{align*}
    \fi
\end{proof}

\bluethree{Theorem \ref{theorem exp} addresses Problem 2 by showing that the values of robust post-dropout policies can be effectively estimated within an exponential confidence bound. }In comparison to Theorem \ref{theorem realization}, the robust value function requires an additional expectation over the dropout probabilities. As this distribution is known \emph{a priori} no error from this operation is introduced. The denominators of the terms vary as the marginalization must now be taken over the whole state space rather than a subset. This result shows that the robust value function estimator may be computed without incurring significantly more error than the estimator for one dropout realization.

\section{Simulations} \label{sec: sims}
\subsection{Robust Policies}
The first experiment, shown in Figure \ref{fig:robust policy}, examines the loss in value incurred when the system undergoes agent dropout. This experiment utilizes the full model to demonstrate the utility of the robust policy before estimation is considered. Four agents are controlled under an optimal policy for 500 time steps. The rewards assigned to the agents are equally weighted. The sample return to the current time is shown in green; under no intervention, this estimate will approximate the true value. At $t=500$, however, the system is disturbed and loses two of the agents. If the CP finds a policy optimal for the post-dropout system, then they can achieve about 42\% of pre-dropout value. If the policy is not adapted and the optimal pre-dropout policy is used, then the return drops to 22\% of the pre-dropout value. This is because the pre-dropout policy does not address the transition matrices marginalizing the removed agents.

The loss in post-dropout value is then compared to the performance of the optimal robust policy, calculated for $\beta_n=0.5$ for all agents. Note that when dropout occurs, the return of the robust policy is 95\% of the optimal pre-dropout value. This shows that while the robust policy is suboptimal, the loss is negligible in this experiment. Bounds for this optimality gap are shown in Figure \ref{fig: gap} for varying values of $\beta$. The true loss is displayed with the theoretical upper bound from Lemma \ref{opt gap}. 

The benefits of the robust policy are shown after dropout occurs; the robust policy yields 37\% of the pre-dropout optimal value, a difference of five points from the optimal post-dropout policy, and an improvement of fifteen points over the optimal pre-dropout policy. This shows that the robust policy is better at automatically adjusting to any realization of the system than policies designed for the pre-dropout system.

To demonstrate the robust policy in a general setting, the results averaged over 1000 randomly generated systems are reported in Figure \ref{fig:all betas}. This experiment considers a system of five agents where $|\X_n| = 2$ and $|\A_n| = 2$, and all possible dropout combinations were considered. For each number of lost agents, the robust policy was found with the corresponding value of $\beta$; for example, all combinations with one lost agent used $\beta_n = 0.2$. The result in red is the loss in value between the optimal post-dropout policy and the optimal pre-dropout policy. The star shows the average loss across all dropout combinations, and the error bars report the minimum and maximum.  The corresponding results in black show the value loss between the optimal post-dropout policy and the optimal robust policy in the post-dropout regime. In comparing the red and black results, note the robust policy performs better in the average for $\beta \in \{0.2, 0.4\}$ and better in the maximum for $\beta\in \{0.2, 0.4, 0.6\}$. This makes sense  as a post-dropout system with fewer dropped agents will be more similar to the pre-dropout system, so the optimal pre-dropout policy should yield less loss. In comparison, the robust policy needs to perform well over the average of all possible dropout combinations, so it may have higher loss for a specific realization for high $\beta$. When the number of dropped agents is high, the pre- and post-dropout systems are less similar, so the robust policy out-performs the pre-dropout policy. 

The last data point is the dot, which is the loss on the value function incurred by controlling the pre-dropout system with the robust policy instead of the optimal pre-dropout policy. There are no error bars as this metric is reported only in the pre-dropout regime. These results show that controlling the existing system with the robust policy results in at most a 10\% loss in performance from the optimum. This experiment shows that the robust policy can be a promising strategy for systems that undergo agent dropout. 

\ifarxiv
\begin{figure}
    \centering
    \includegraphics[width = 0.6\linewidth]{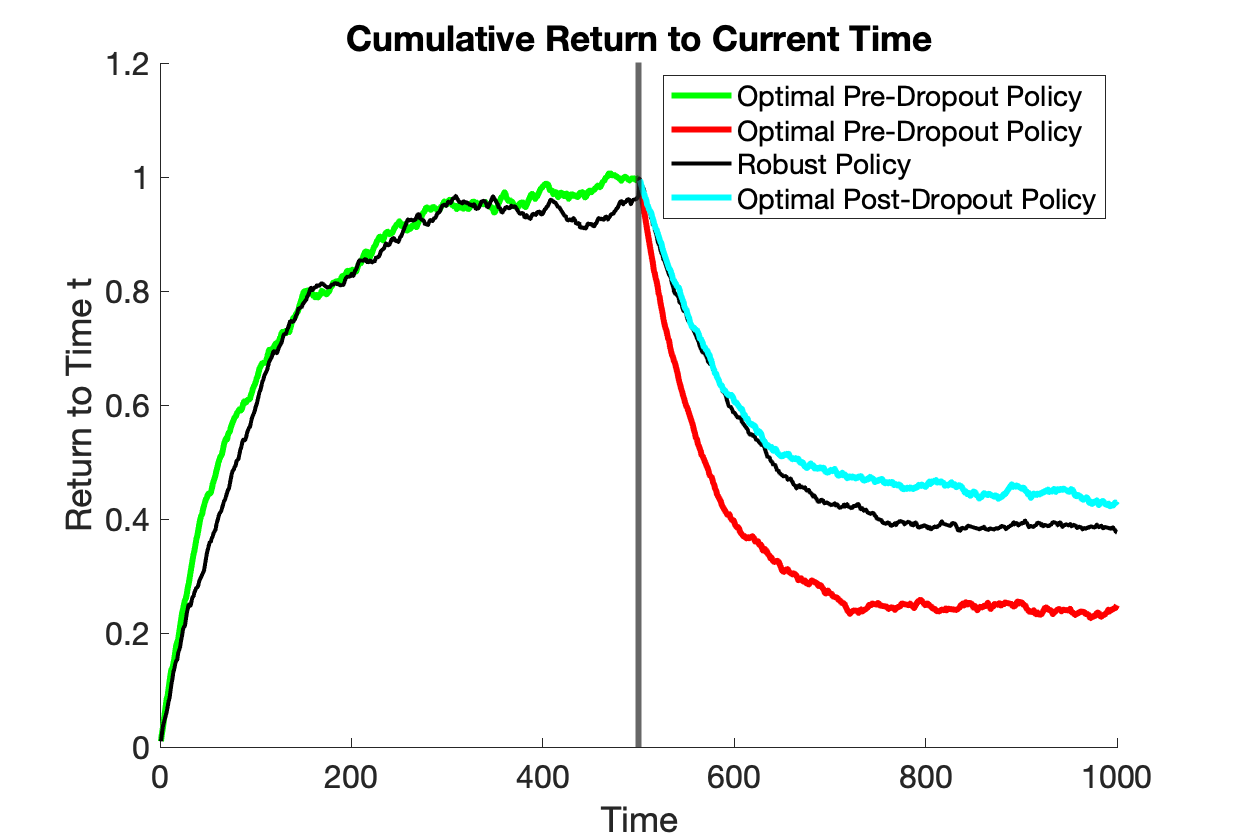}
    \caption{Experimental comparison of various policies on the pre- and post-dropout system. \blue{$N=4$, $|\X_n| = 2$, $|\mathcal{A}_n|=2$. Note that the robust policy performs almost as well as the optimal pre-dropout policy for $t<500$, and is much better than the optimal pre-dropout policy for $t>500$. In addition, the central planner can pre-compute the robust policy with policy IS.}}
    \label{fig:robust policy}
\end{figure}
\fi

\ifarxiv
\begin{figure}
    \centering
    \includegraphics[width = 0.6\linewidth]{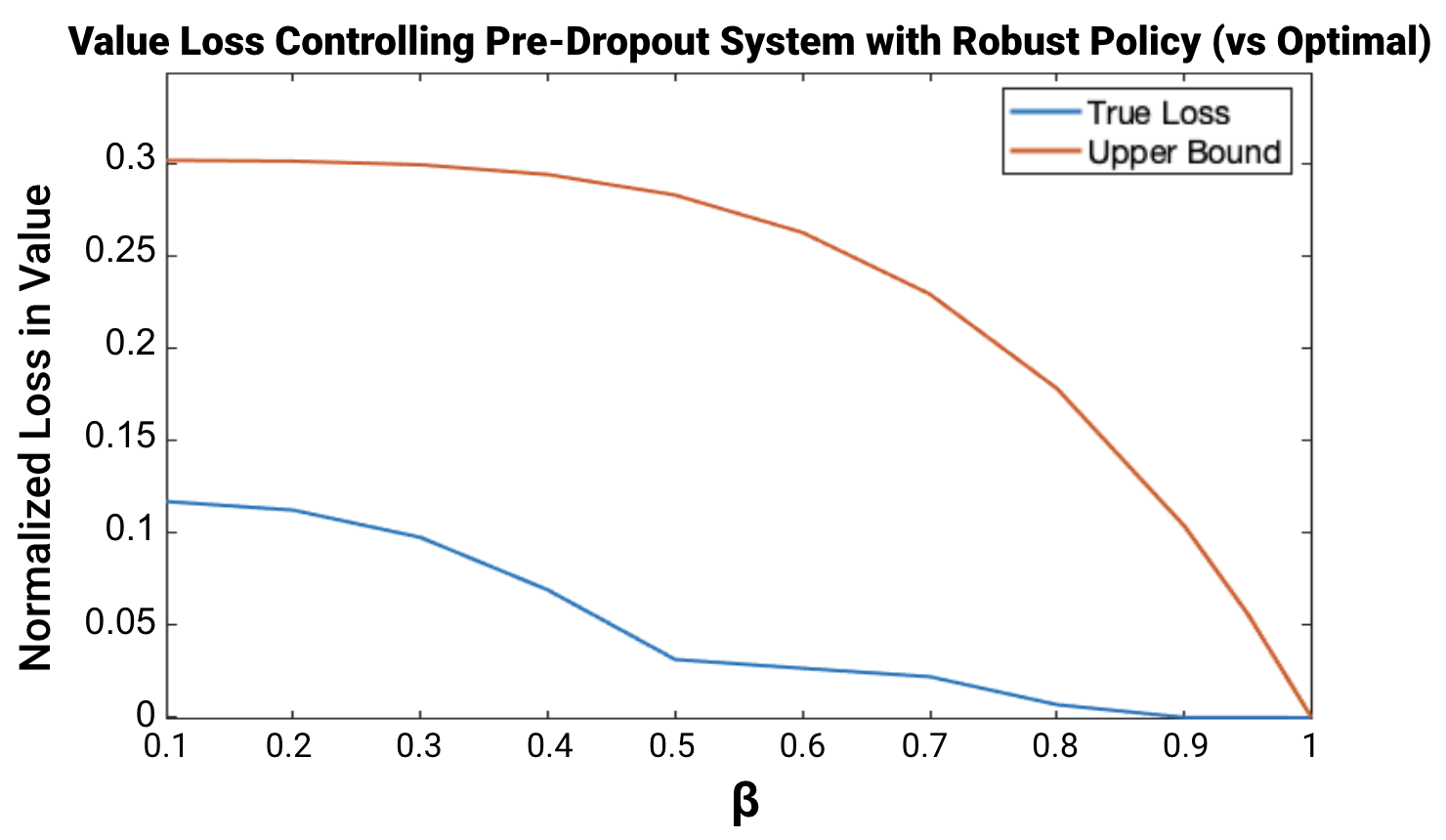}
    \caption{\blue{Experimental demonstration of the optimality gap and associated upper bound (Lemma \ref{opt gap}). $N=4$, $|\X_n| = 3$, $|\A_n| = 3$, and the rewards were assigned such that $|r_n(x_n,\alpha_n|w_n=1)| \leq 1$ and $r_n(x_n,\alpha_n|w_n=0) = 0$.}}
    \label{fig: gap}
\end{figure}
\fi

\ifarxiv
\begin{figure}
    \centering
    \includegraphics[width = 0.8\linewidth]{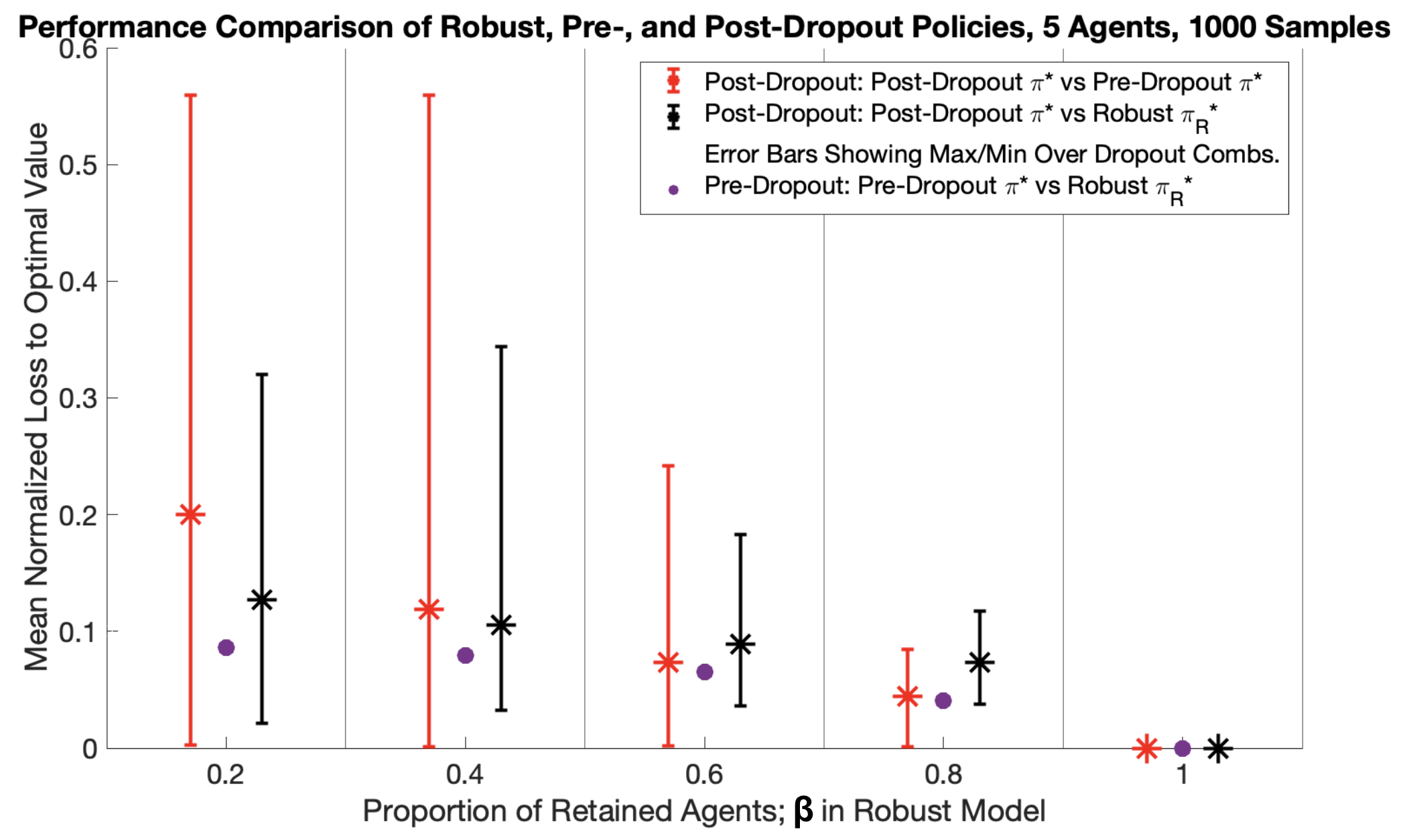}
    \caption{\blue{Empirical losses produced by the robust policy and optimal pre-dropout policy on post-dropout system. $N=5$, $|\X_n| = 2$, $|\A_n|=2$, averaged across all dropout combinations, and 1000 experiments. Error bars show max/min loss over the dropout combinations. The robust policy performs better on average than the pre-dropout policy when for less than half of the agents, and it performs better in the maximum for $\beta\in\{0.2, 0.4, 0.6\}$.}}
    \label{fig:all betas}
\end{figure}
\fi

\subsection{Importance Sampling}
\blue{The next experiment looked at estimating robust policies from data. The CP gathered data under $\pi^*$ optimal for $W = \mathbf{1}$, and then used the proposed policy IS method to evaluate a candidate robust policy. The policy IS routine was implemented with a first-visit doubly \bluetoo{robust} estimator \cite{jiang2016doubly}, $|D|=100$,  $H=500$, and $H_{\mu} = 5000$. The results are in Figure \ref{fig:need is}.}

\blue{The results demonstrate how the candidate robust policy can be estimated accurately. The results are compared to the pre-dropout system to show that the robust value can be estimated while controlling the existing system with its optimal policy. The horizontal dashed line shows the true value of the policy normalized to one, and the solid green horizontal line shows the estimated value of the same policy as found by the policy IS routine. The vertical red line shows that if the candidate policy was evaluated by directly controlling the existing system and observing the sample return, then 300 time steps would be needed to reach 95\% of the true value. Works such as \cite{thomas2015high} have investigated how to perform the policy search and improvement steps.} 

\ifarxiv
\begin{figure}
    \centering
    \includegraphics[width=0.7\linewidth]{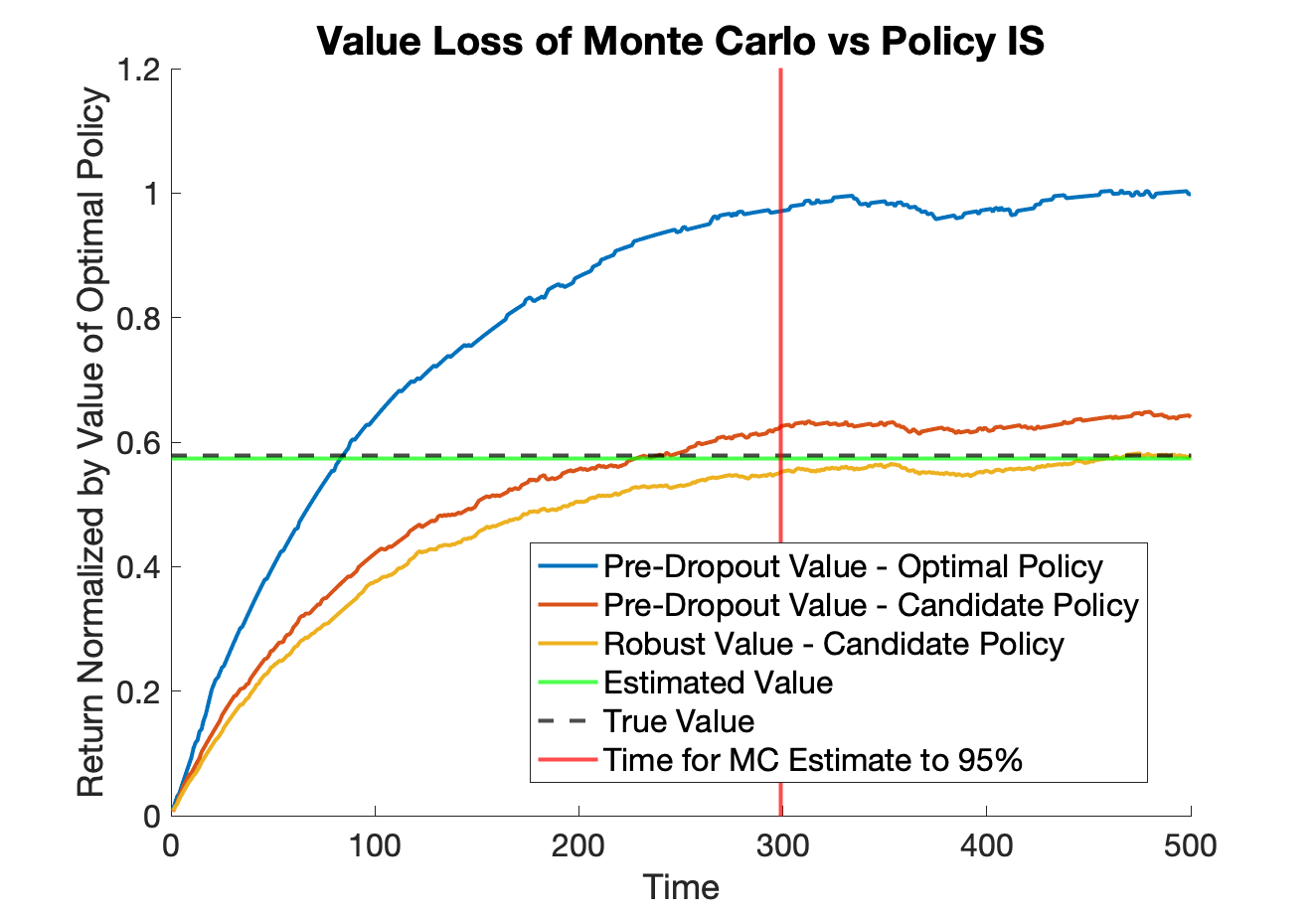}
    \caption{\blue{Performance of policy IS to evaluate a candidate policy. $N=5$, $|\X_n|=2$, $|\mathcal{A}_n|=2$, $|D|=100$,  $H=500$, and $H_{\mu} = 5000$. Note that the policy IS estimator approximates the true robust value. The candidate policy value of the robust MDP is close to the pre-dropout MDP. This particular candidate policy on the robust system achieves  $\sim$60\% of the optimal pre-dropout MDP value. }}
    \label{fig:need is}
\end{figure}
\fi

\section{Conclusion}
\blue{In this paper we consider a multi-agent MDP that undergoes agent dropout. For both computational concerns and to find good control policies before dropout occurs, we adopt a stochastic agent dropout model and propose a robust value function that optimizes for the expected system composition.  This model produces two key takeaways: the development of (1) the robust policy framework, which enables the attainment of good value across different dropout realizations, and (2) a statistical method for estimating such policies before dropout occurs. To complete these objectives, we propose a policy IS method that can be computed from pre-dropout observations; this means that potential robust or post-dropout policies can be evaluated before dropout occurs, which means that good control policies can be found before dropout occurs. We prove that this method converges with high probability, up to a possible bias term produced from importance sampling. Experiments validate that the robust policy can perform well if dropout occurs.}

\blue{Future work can consider the composition of $\beta$ to decide when it is beneficial to use the robust policy, which would depend on both the expected value and the needed computation. Second, the assumption of $\beta$ known \emph{a priori} could be relaxed, necessitating estimation of both the policy and the expected system. In addition, the results could be expanded to consider nodes added to the system, which will require exploration time to learn the unknown behavior.}

\bibliographystyle{IEEEtran}
\bibliography{IEEEabrv,main}

\begin{thebibliography}{10}
\providecommand{\url}[1]{#1}
\csname url@samestyle\endcsname
\providecommand{\newblock}{\relax}
\providecommand{\bibinfo}[2]{#2}
\providecommand{\BIBentrySTDinterwordspacing}{\spaceskip=0pt\relax}
\providecommand{\BIBentryALTinterwordstretchfactor}{4}
\providecommand{\BIBentryALTinterwordspacing}{\spaceskip=\fontdimen2\font plus
\BIBentryALTinterwordstretchfactor\fontdimen3\font minus \fontdimen4\font\relax}
\providecommand{\BIBforeignlanguage}[2]{{%
\expandafter\ifx\csname l@#1\endcsname\relax
\typeout{** WARNING: IEEEtran.bst: No hyphenation pattern has been}%
\typeout{** loaded for the language `#1'. Using the pattern for}%
\typeout{** the default language instead.}%
\else
\language=\csname l@#1\endcsname
\fi
#2}}
\providecommand{\BIBdecl}{\relax}
\BIBdecl

\bibitem{le2022socialbots}
T.~Le, L.~Tran-Thanh, and D.~Lee, ``Socialbots on fire: Modeling adversarial behaviors of socialbots via multi-agent hierarchical reinforcement learning,'' in \emph{Proceedings of the ACM Web Conference 2022}, 2022, pp. 545--554.

\bibitem{fang2021multi}
X.~Fang, Q.~Zhao, J.~Wang, Y.~Han, and Y.~Li, ``Multi-agent deep reinforcement learning for distributed energy management and strategy optimization of microgrid market,'' \emph{Sustainable cities and society}, vol.~74, p. 103163, 2021.

\bibitem{chen2021graph}
S.~Chen, J.~Dong, P.~Ha, Y.~Li, and S.~Labi, ``Graph neural network and reinforcement learning for multi-agent cooperative control of connected autonomous vehicles,'' \emph{Computer-Aided Civil and Infrastructure Engineering}, vol.~36, no.~7, pp. 838--857, 2021.

\bibitem{zhu2023neural}
Y.~Zhu, Z.~Wang, H.~Liang, and C.~K. Ahn, ``Neural-network-based predefined-time adaptive consensus in nonlinear multi-agent systems with switching topologies,'' \emph{IEEE Transactions on Neural Networks and Learning Systems}, 2023.

\bibitem{zhang2020robust}
H.~Zhang, H.~Chen, C.~Xiao, B.~Li, M.~Liu, D.~Boning, and C.-J. Hsieh, ``Robust deep reinforcement learning against adversarial perturbations on state observations,'' \emph{Advances in Neural Information Processing Systems}, vol.~33, pp. 21\,024--21\,037, 2020.

\bibitem{carmona2023model}
R.~Carmona, M.~Lauri{\`e}re, and Z.~Tan, ``Model-free mean-field reinforcement learning: mean-field mdp and mean-field q-learning,'' \emph{The Annals of Applied Probability}, vol.~33, no.~6B, pp. 5334--5381, 2023.

\bibitem{shapley1953stochastic}
L.~S. Shapley, ``Stochastic games,'' \emph{Proceedings of the national academy of sciences}, vol.~39, no.~10, pp. 1095--1100, 1953.

\bibitem{yang2020overview}
Y.~Yang and J.~Wang, ``An overview of multi-agent reinforcement learning from game theoretical perspective,'' \emph{arXiv preprint arXiv:2011.00583}, 2020.

\bibitem{Fiscko2019ControlOP}
C.~Fiscko, B.~Swenson, S.~Kar, and B.~Sinopoli, ``Control of parametric games,'' \emph{2019 18th European Control Conference (ECC)}, pp. 1036--1042, 2019.

\bibitem{osband2014near}
I.~Osband and B.~Van~Roy, ``Near-optimal reinforcement learning in factored mdps,'' \emph{Advances in Neural Information Processing Systems}, vol.~27, 2014.

\bibitem{becker2003transition}
R.~Becker, S.~Zilberstein, V.~Lesser, and C.~V. Goldman, ``Transition-independent decentralized markov decision processes,'' in \emph{Proceedings of the second international joint conference on Autonomous agents and multiagent systems}, 2003, pp. 41--48.

\bibitem{fiscko2022cluster}
\BIBentryALTinterwordspacing
C.~Fiscko, S.~Kar, and B.~Sinopoli, ``Clustered control of transition-independent mdps,'' \emph{arXiv preprint arXiv:2207.05224}, 2022. [Online]. Available: \url{https://arxiv.org/abs/2207.05224}
\BIBentrySTDinterwordspacing

\bibitem{ornik2019learning}
M.~Ornik and U.~Topcu, ``Learning and planning for time-varying mdps using maximum likelihood estimation,'' \emph{arXiv preprint arXiv:1911.12976}, 2019.

\bibitem{lecarpentier2019non}
E.~Lecarpentier and E.~Rachelson, ``Non-stationary markov decision processes, a worst-case approach using model-based reinforcement learning,'' \emph{Advances in neural information processing systems}, vol.~32, 2019.

\bibitem{cheung2020reinforcement}
W.~C. Cheung, D.~Simchi-Levi, and R.~Zhu, ``Reinforcement learning for non-stationary markov decision processes: The blessing of (more) optimism,'' in \emph{International Conference on Machine Learning}.\hskip 1em plus 0.5em minus 0.4em\relax PMLR, 2020, pp. 1843--1854.

\bibitem{summers2009addressing}
T.~H. Summers, C.~Yu, and B.~D. Anderson, ``Addressing agent loss in vehicle formations and sensor networks,'' \emph{International Journal of Robust and Nonlinear Control: IFAC-Affiliated Journal}, vol.~19, no.~15, pp. 1673--1696, 2009.

\bibitem{gasparri2017bounded}
A.~Gasparri, L.~Sabattini, and G.~Ulivi, ``Bounded control law for global connectivity maintenance in cooperative multirobot systems,'' \emph{IEEE Transactions on Robotics}, vol.~33, no.~3, pp. 700--717, 2017.

\bibitem{kar2013cal}
S.~Kar, J.~M. Moura, and H.~V. Poor, ``Qd-learning: A collaborative distributed strategy for multi-agent reinforcement learning through consensus+ innovations,'' \emph{IEEE Transactions on Signal Processing}, vol.~61, no.~7, pp. 1848--1862, 2013.

\bibitem{zhang2018fully}
K.~Zhang, Z.~Yang, H.~Liu, T.~Zhang, and T.~Basar, ``Fully decentralized multi-agent reinforcement learning with networked agents,'' in \emph{International Conference on Machine Learning}.\hskip 1em plus 0.5em minus 0.4em\relax PMLR, 2018, pp. 5872--5881.

\bibitem{thomas2015high}
P.~Thomas, G.~Theocharous, and M.~Ghavamzadeh, ``High confidence policy improvement,'' in \emph{International Conference on Machine Learning}.\hskip 1em plus 0.5em minus 0.4em\relax PMLR, 2015, pp. 2380--2388.

\bibitem{fiscko2022confident}
C.~Fiscko, S.~Kar, and B.~Sinopoli, ``On confident policy evaluation for factored markov decision processes with node dropouts,'' in \emph{2022 IEEE 61st Conference on Decision and Control (CDC)}.\hskip 1em plus 0.5em minus 0.4em\relax IEEE, 2022, pp. 2857--2863.

\bibitem{jiang2016doubly}
N.~Jiang and L.~Li, ``Doubly robust off-policy value evaluation for reinforcement learning,'' in \emph{International Conference on Machine Learning}.\hskip 1em plus 0.5em minus 0.4em\relax PMLR, 2016, pp. 652--661.

\bibitem{osborne1994course}
M.~J. Osborne and A.~Rubinstein, \emph{A course in game theory}.\hskip 1em plus 0.5em minus 0.4em\relax MIT press, 1994.

\bibitem{guestrin2003efficient}
C.~Guestrin, D.~Koller, R.~Parr, and S.~Venkataraman, ``Efficient solution algorithms for factored mdps,'' \emph{Journal of Artificial Intelligence Research}, vol.~19, pp. 399--468, 2003.

\bibitem{bertsekas1995dynamic}
D.~P. Bertsekas, D.~P. Bertsekas, D.~P. Bertsekas, and D.~P. Bertsekas, \emph{Dynamic programming and optimal control}.\hskip 1em plus 0.5em minus 0.4em\relax Athena scientific Belmont, MA, 1995, vol.~1, no.~2.

\bibitem{wolfer2020mixing}
G.~Wolfer, ``Mixing time estimation in ergodic markov chains from a single trajectory with contraction methods,'' in \emph{Algorithmic Learning Theory}.\hskip 1em plus 0.5em minus 0.4em\relax PMLR, 2020, pp. 890--905.

\bibitem{paulin2015concentration}
D.~Paulin, ``Concentration inequalities for markov chains by marton couplings and spectral methods,'' \emph{Electronic Journal of Probability}, vol.~20, pp. 1--32, 2015.

\bibitem{thomas2016data}
P.~Thomas and E.~Brunskill, ``Data-efficient off-policy policy evaluation for reinforcement learning,'' in \emph{International Conference on Machine Learning}.\hskip 1em plus 0.5em minus 0.4em\relax PMLR, 2016, pp. 2139--2148.

\end{thebibliography}

\end{document}